\documentclass[acmtog,nonacm]{acmart}

\acmConference{}{}{}

\usepackage[ruled]{algorithm2e} 

\SetAlFnt{\small}
\SetAlCapFnt{\small}
\SetAlCapNameFnt{\small}
\SetAlCapHSkip{0pt}

\usepackage{multirow}

\usepackage{xspace}
\usepackage{enumitem}

\newcommand{\oldversion } [1] { }
\newcommand{\ignorethis } [1] { }

\newcommand{\enrico } [1] {\textcolor{purple}{\textbf{e:} {\slshape #1}}}

\newcommand{\fignum     } [1] {\ref{#1}}

\newcommand{\cfignum}[2]{\IfRefUndefinedExpandable{#1}{#2}{\fignum{#1}}}

\newcommand{\Reals      }     {{\textrm{I\kern-0.18em R}}}

\newcommand{\change     } [1] {\mbox{{\footnotesize $\Delta$} \kern-3pt}#1}

\newcommand{\argmin     } [1] {{\underset{#1}{\operatorname{argmin}}}}

\theoremstyle{plain}
\newtheorem{theorem}{Theorem}[section]
\newtheorem{prop}[theorem]{Proposition}

\hypersetup{draft}
\begin{document}

\title{b/Surf: Interactive B\'ezier Splines on Surfaces}

\author{Claudio Mancinelli}
\affiliation{\institution{University of Genoa} 
\department{Department of Informatics, Bioengineering, Robotics and Systems Engineering}
}
\authornote{Joint first authors}

\author{Giacomo Nazzaro}
\affiliation{\institution{Sapienza University of Rome}
\department{Department of Computer Science}
}
\authornotemark[1]

\author{Fabio Pellacini}
\affiliation{\institution{Sapienza University of Rome}
\department{Department of Computer Science}
}

\author{Enrico  Puppo}
\affiliation{\institution{University of Genoa}
\department{Department of Informatics, Bioengineering, Robotics and Systems Engineering}
}



\begin{abstract}
\label{sec:abstract}
Bézier curves provide the basic building blocks of graphic design in 2D.
In this paper, we port Bézier curves to manifolds. We support the interactive 
drawing and editing of Bézier splines on manifold meshes with 
millions of triangles, by relying on just repeated manifold averages.  
We show that direct extensions of the De Casteljau and Bernstein evaluation 
algorithms to the manifold setting are fragile, and prone to discontinuities 
when control polygons become large. Conversely, approaches based on subdivision 
are robust and can be implemented efficiently.
We define Bézier curves on manifolds, by extending both the recursive 
De Casteljau bisection and a new open-uniform Lane-Riesenfeld subdivision scheme, 
which provide curves with different degrees of smoothness. 
For both schemes, we present algorithms for curve tracing, point evaluation, 
and point insertion.
We test our algorithms for robustness and performance on all watertight, 
manifold, models from the Thingi10k repository, without any pre-processing and 
with random control points.
For interactive editing, we port all the basic user interface interactions 
found in 2D tools directly to the mesh. We also support mapping complex 
SVG drawings to the mesh and their interactive editing.
\end{abstract}

\ccsdesc[500]{Computing methodologies~Mesh models}
\ccsdesc[500]{Computing methodologies~Parametric curve and surface models}
\ccsdesc[300]{Computing methodologies~Graphics systems and interfaces}

\keywords{geometric meshes, spline curves, user interfaces, geometry processing}

\begin{teaserfigure}
  \centering
  \includegraphics[width=0.9\textwidth]{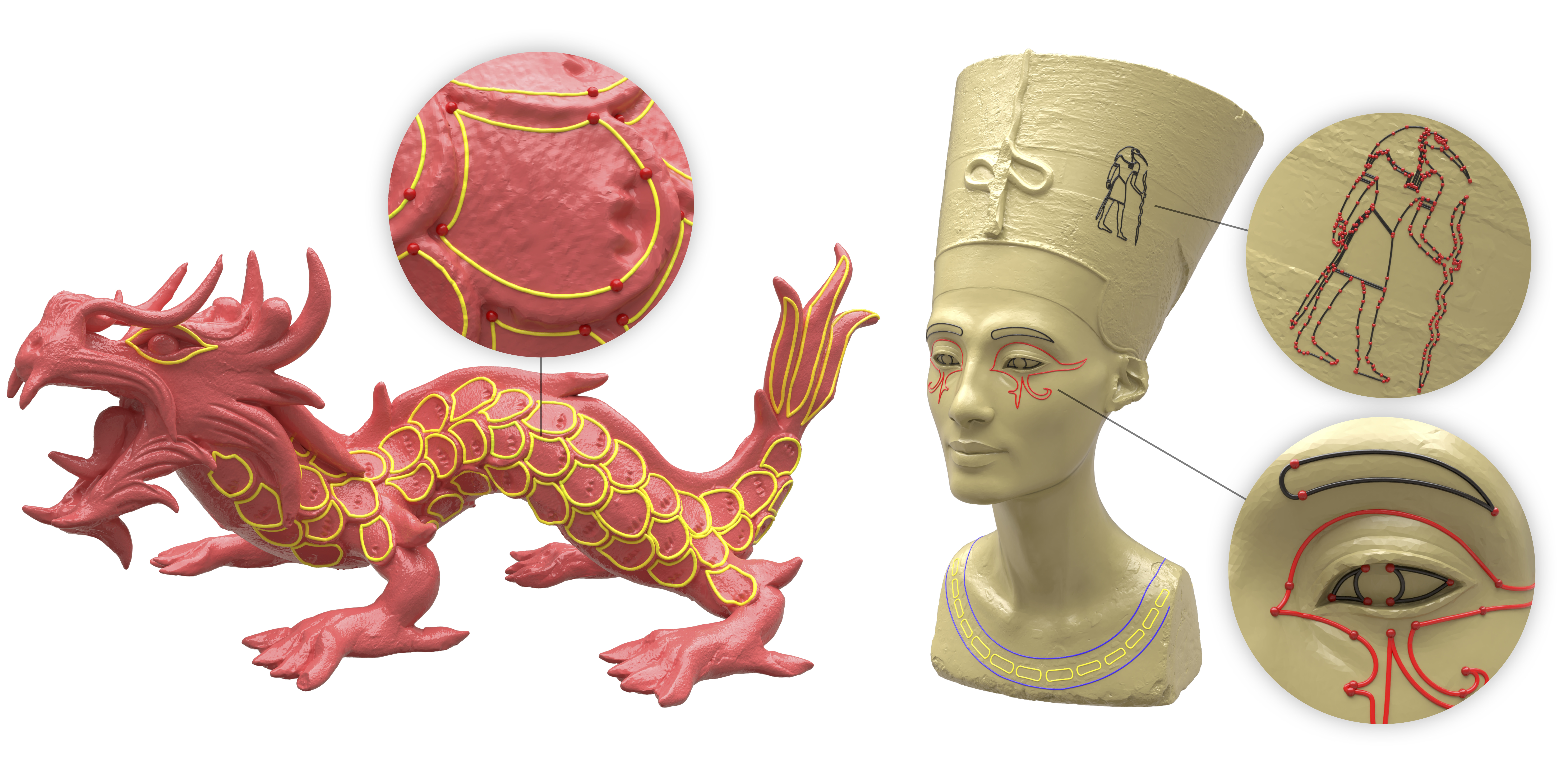}
  \caption{We propose algorithms to interactively edit Bézier 
  splines on large meshes, including curve editing, curve transformations and 
  import and editing complex SVG drawings. 
  All computations occur in the intrinsic geodesic metric of the surface. 
  All splines in this figure have been drawn interactively. 
  Control points and tangents of curves under editing are shown in the zoomed insets. 
  Asian Dragon $\sim$7.2M triangles; 
  Nefertiti $\sim$500K triangles.} 
  \label{fig:teaser}
  \Description{Image}
\end{teaserfigure}

\maketitle

\newcommand{\reswidth}[0]{7in}


\section{Introduction}
\label{sec:intro}

Vector graphics in 2D is consolidated since decades, as is supported in many 
design applications, such as Adobe Illustrator \cite{adobeillustrator}, and 
languages, like Scalable Vector Graphics (SVG) \cite{SVG}.
Bézier curves are the building blocks of most vector graphics packages, since
most other primitives 
can be converted into \emph{B\'ezier splines} (chains of Bézier curves) and edited as such \cite{Farin2001}.

In many design applications, it would be beneficial to edit vector graphics 
directly on surfaces, instead of relying on parametrization or projections
that have inherent distortions \cite{Poerner:2018gp}. 
Yet, bringing vector graphics to surfaces is all but trivial, since basic rules 
of Euclidean geometry do not hold under the geodesic metric on manifolds; and 
distances, shortest and straightest paths cannot be computed in closed form. 
In particular, in spite of several attempts to define curves under the 
geodesic metric, a complete computational framework that supports their 
practical usage in an interactive design setting is still missing.

In this work, we port Bézier curves on surfaces. We begin by analyzing the 
extensions of Bézier curves to the manifold setting, which are obtained by 
replacing straight lines with geodesic lines, and affine averages with the 
Riemannian center of mass.
We show that direct approaches to curve evaluation are fragile, and may 
generate discontinuous curves. On the contrary, approaches based on subdivision 
and repeated averages are robust. In particular, the scheme based on recursive 
De Casteljau bisection lends itself to an efficient implementation. 
Inspired by recent theoretical results \cite{Duchamp:2018eu,Dyn:2019cl}, 
we propose a new scheme based on an open-uniform Lane-Riesenfeld subdivision, 
which guarantees higher smoothness, and can also be implemented efficiently.  

Focusing on such two schemes, we design algorithms for curve tracing, 
point evaluation, and point insertion. Our implementation rests on light 
data structures and just a few basic tools for geodesic computations.
The proposed algorithms remain interactive on meshes made of millions of triangles 
such as the ones shown in Fig.~\ref{fig:teaser}.
To assess the robustness and performance of these algorithms, we trace curves on the more than 
five-thousands watertight, manifold, meshes of the Thingi10k repository 
\cite{Thingi10K} with many randomly generated control polygons. 
We show that our algorithms handle well all cases. 

\begin{figure}[t!]
  \centering
    \includegraphics[width=0.9\columnwidth]{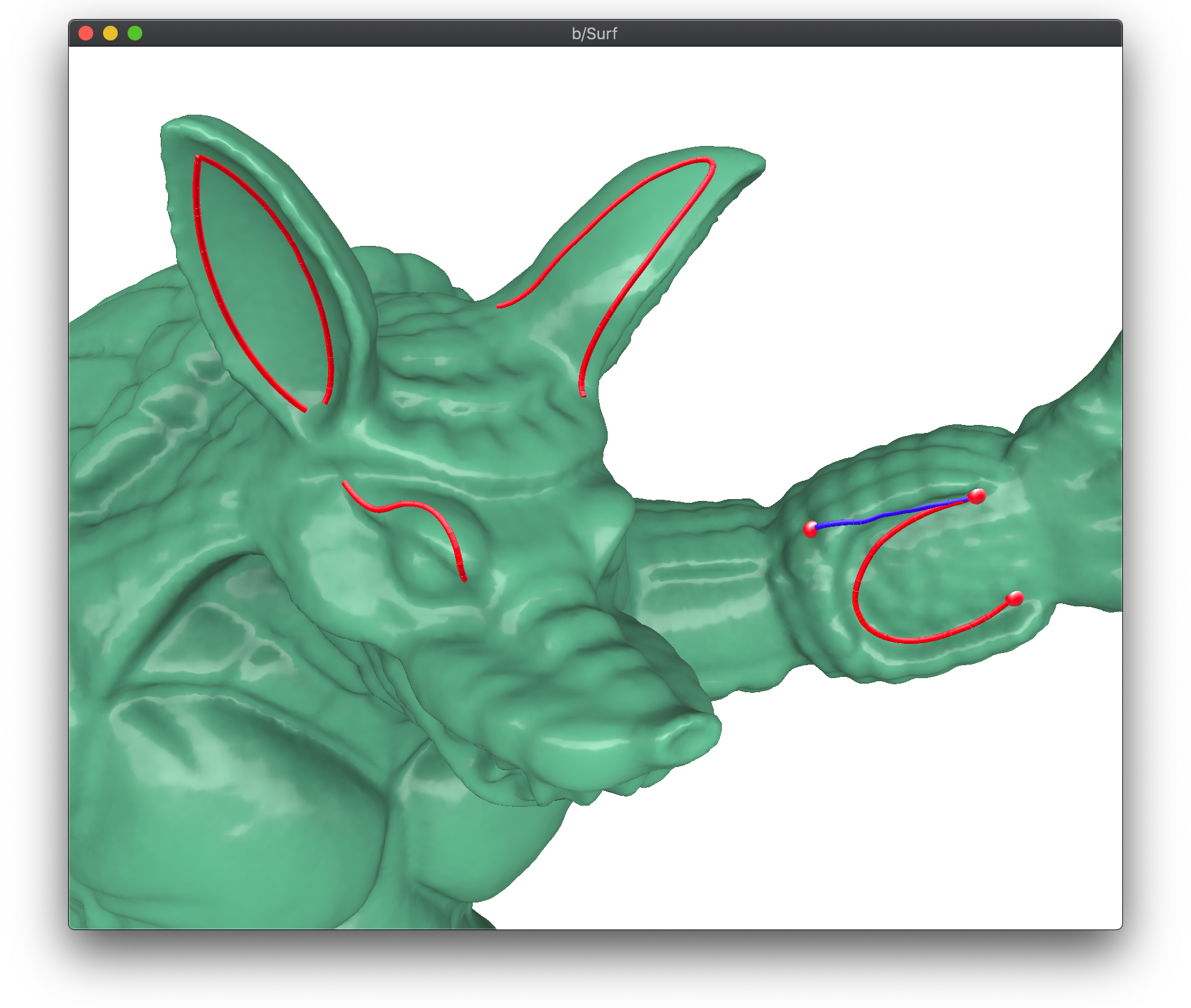}
  \caption{The GUI of our system. 
  Curves on the eye and on the arm consist of a single cubic segment each,  
  while the two splines on the ears consist of two cubic segments each, 
  with a sharp corner to the left and smooth junction to the right.   
  A control tangent (in blue) is depicted on the curve under editing. 
  } 
  \label{fig:simple}
  \Description{Image}
\end{figure}

We integrate curve tracing and point insertion algorithms in a user interface 
prototype that supports the interactive design of Bèzier splines 
on manifold meshes. We support all basic operations that 2D editors have, 
including: click-and-drag of control points and tangents; point insertion and 
deletion; and translation, rotation and scaling of curves. We also support
mapping of 2D drawings onto the surface.
Fig.~\ref{fig:simple} shows the interface of our system with some simple 
curves traced on a model. The supplemental video shows a full editing session.

All together, this work advances the state of the art with five main 
contributions:
\begin{itemize}
\item We present a critical analysis of several definitions of Bézier curves 
in the manifold setting. We demonstrate the limitations of direct methods, 
and show that subdivision schemes are always well behaved.
\item We provide algorithms for curve tracing, point evaluation, 
and point insertion, for a recursive De Casteljau scheme, and a 
novel open-uniform Lane-Riesenfeld scheme that warrants higher smoothness 
than previous definitions.
\item We provide a very fast implementation of such algorithms
that runs at interactive rates for meshes of millions of triangles on 
single-core commodity CPUs. 
Upon publication, we will release all source code with a 
permissive license.
\item We show that our algorithms are robust and interactive with a large test 
on Thingi10k models \cite{Thingi10K}, and we compare them with state-of-the-art solutions \cite{geometrycentral,Panozzo:2013eh}.
\item We develop a prototype system for the interactive design of 
Bézier splines, supporting all operations commonly found in 2D editors.
Our editor remains interactive with meshes of millions of triangles.
We also show how to port existing 2D drawings onto surfaces.
\end{itemize}


\section{Related work}
\label{sec:related}

The design of spline curves on manifolds has been addressed by several authors, 
both from a mathematical and from a computational perspective. 
We review only methods addressing general surfaces. 

A traditional approach to circumvent the problems of the Riemannian metric 
consists of linearizing the manifold domain via parametrization, 
designing curves in the parametric plane, and mapping the result to the surface.
Parametrization introduces seams, and drawing lines across them becomes problematic. 
Moreover, distortions induced by parametrizations are hard to predict and control.
The exponential map can provide a local parametrization on the fly for the 
region of interest \cite{Biermann:2002bi,Herholz:2019bk,Schmidt:2013bj,SunI3D:2013,Schmidt:2006tt}. 
However, its radius of injectivity is small in regions of high curvature, 
while control polygons and curves may extend over large regions.
Even curves as simple as the ones depicted in Fig.~\ref{fig:loop} may be hard 
to control using either local or global parametrizations.

\begin{figure}[t!]
  \centering
  \includegraphics[width=\columnwidth]{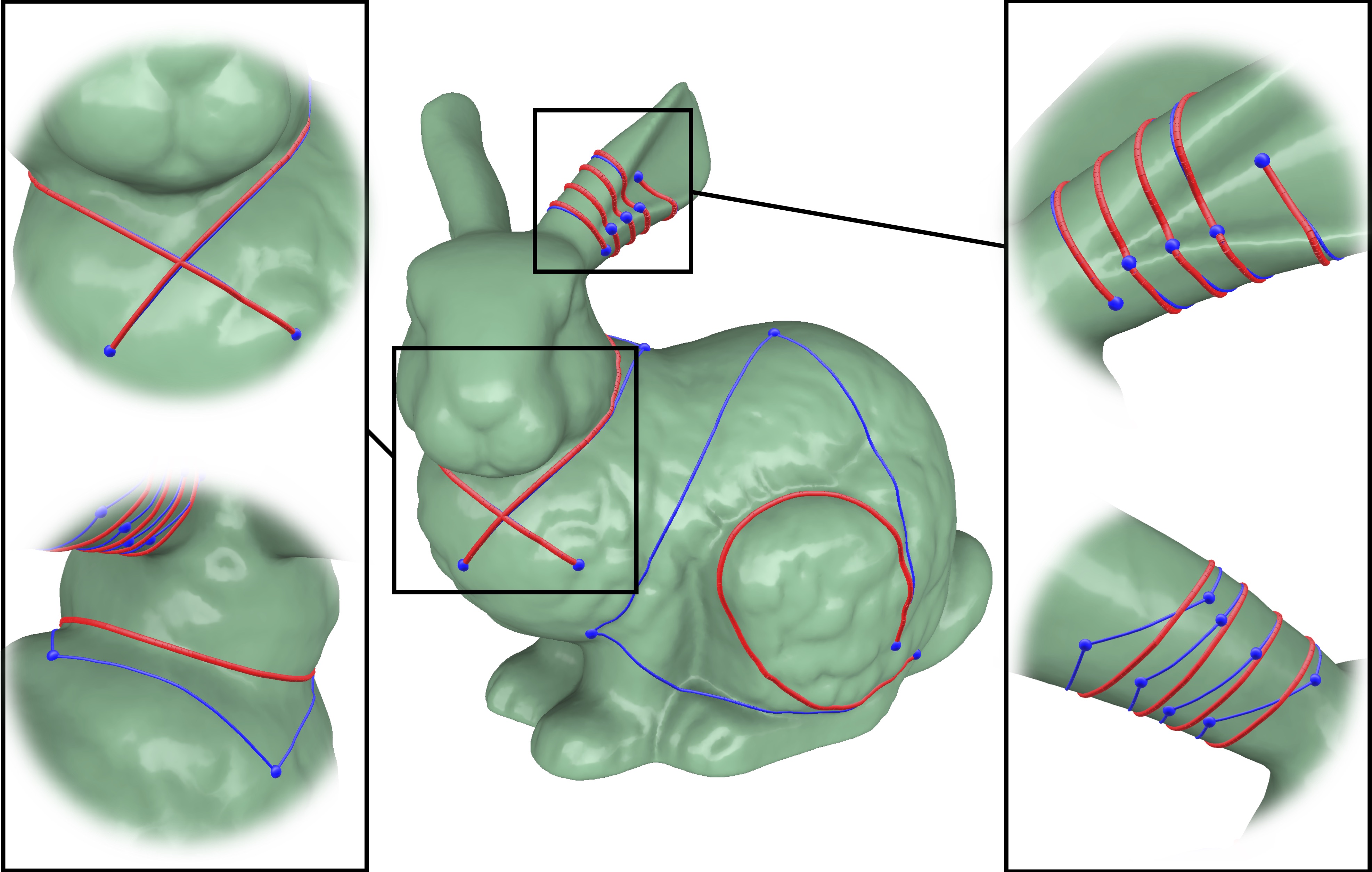}
  \caption{
  Curves that wind about the object or require large control polygons may be 
  challenging to draw with an approach based on parametrization. 
  The collar and the curl consist each of a single cubic segment, while
  the spiral is a spline of four segments joined with smooth ($C^1$) continuity. 
  Control polygons are depicted in blue.
}
  \label{fig:loop}
  \Description{Image}
\end{figure}

Another approach consists of relaxing the manifold constraint, 
resolving the problem in a space that admits computations in closed form, 
and projecting the result back to the surface.  
\cite{Wallner:2006eq} computes curves in 3D space and projects them 
to the nearest points on surface. 
\cite{Panozzo:2013eh} uses an embedding in a higher-dimensional Euclidean space, 
followed by Phong projection. 
These methods may support user interaction, but they provide only approximate 
results, are prone to artifacts, and are hard to scale to large meshes. 
In Section~\ref{sub:flipOut}, we further discuss the method of \cite{Panozzo:2013eh} 
and compare its results and performances with our method. 

The design of curves can also be addressed as an optimization problem 
in a variational setting. 
\cite{Noakes:wo} and \cite{Camarinha:1995} provide the basic variational 
theory of splines on manifolds.
This approach is adopted in several other papers 
\cite{Arnould:2015ce,Gousenbourger:2014ft,Gousenbourger:2018js,Hofer:2004fc,Jin:2019,Pottmann:2005ci,Samir:2011bu}. 
While most such works do not address implementation and performance,
\cite{Hofer:2004fc} and \cite{Jin:2019} eventually resort to projection methods.
Overall, the variational approach is too computationally expensive to support 
user interaction on large meshes.
Moreover, these curves are harder to control interactively than traditional 
Bézier splines.

Concerning the specific case of Bézier curves, \cite{Park:1995hu} first extended
the De Casteljau algorithm to Riemannian manifolds, expressing geodesic lines 
of control polygons through the exponential map, without further developing 
the computational details.
Later on, the De Casteljau algorithm on surfaces has been explored by several 
other authors \cite{Gousenbourger:2018js,Lin:2001ua,Morera:2008jk,NavaYazdani:2013dh,Popiel:2007bt}.
Among these, \cite{Morera:2008jk} extends the recursive De Casteljau bisection,
and \cite{geometrycentral} achieves performance on the same algorithm, by using 
a fast method for evaluating locally shortest geodesic paths \cite{sharp2020flipout}.
We adopt the same recursive structure of \cite{Morera:2008jk} 
for curve tracing with the recursive De Casteljau bisection.
In Section~\ref{sub:flipOut}, we further discuss the method of \cite{geometrycentral,sharp2020flipout} 
and compare their results and performances with our method. 
\cite{Absil:2016bd} defines Bézier curves both with the De Casteljau algorithm 
and with the Riemannian center of mass, and show that they may produce different results.

Several authors have investigated the theoretical aspects of the subdivision 
approach to splines in the manifold setting. 
We refer to \cite{Wallner:2020} for a detailed analysis on the subject; here, we report just
the ones on which our algorithms rely.
\cite{Noakes:1998bo} proves that the recursive De Casteljau bisection converges 
and produces a $C^1$ curve in the cubic case, and \cite{Noakes:1999er} shows 
that this is also true for the quadratic case.
Most recent results \cite{Duchamp:2018eu,Dyn:2017fr,Dyn:2019cl} focus on 
Lane-Riesenfeld schemes and show that a scheme of order $k$ is convergent and 
$C^k$ in the manifold and functional settings.
These latter works motivate our approach to the open-uniform Lane-Riesenfeld 
subdivision.

\section{Bézier Curves on Manifolds}
\label{sec:method}

In this section, we consider different constructions for Bézier curves, 
all of which produce the same curves in the Euclidean setting, 
and we analyze their possible extensions to the manifold setting.
We show that the classical De Casteljau and Bernstein evaluation algorithms 
may fail in the manifold setting, while methods based on subdivision guarantee 
different degrees of smoothness.

Here, we only provide the basics of each construction, and refer the readers to 
\cite{Farin2001,Salomon:2006wp} for further details in the Euclidean setting.
We assume that readers are familiar with geodesic lines and shortest paths on 
manifolds, and provide a brief introduction on the subject in Appendix 
\ref{app:geodesic}. 


\subsection{Preliminaries and notations}
\label{sub:prelim}


In the Euclidean setting, a Bézier curve is a polynomial parametric function of degree $k$ 
$$\mathbf{b}^k : [0,1]\longrightarrow \mathbb{R}^d,$$
which is defined by means of a \emph{control polygon} $\Pi_k=(P_0,\ldots, P_k)$, where all $P_i\in\mathbb{R}^d$.
Curve $\mathbf{b}^k$ interpolates points $P_0$ and $P_k$, and it is tangent to $\Pi_k$ at them. 
When there is no ambiguity, we will omit the subscript $k$ and will denote the control polygon simply by $\Pi$.

All constructions of Bézier curves in the Euclidean setting rely on the 
computation of \emph{affine averages} of points of the form
\begin{equation}
\label{eq:affine}
\bar{P} = \sum_{i=0}^h w_i P_i
\end{equation}
where the $w_i$ are non-negative weights satisfying the partition of unity.
For $h=1$, the affine average reduces to the linear interpolation 
\begin{equation}
\label{eq:affine2}
\bar{P}=(1-w) P + w Q.
\end{equation}
By analogy with the Euclidean setting, a control polygon $\Pi_k$ in the
manifold setting consists of a polyline of shortest geodesic paths,
connecting the control points that lie on a complete manifold $M$.

Affine averages are not available on manifolds, but they can be substituted 
with the Riemannian center of mass as introduced by \cite{Grove1973,Karcher:1977ei}. 
Given points $P_0, \ldots , P_h \in M$ and weights $w_0, \ldots , w_h$, 
as before, their \emph{Riemanninan Center of Mass} (RCM) on $M$ 
is given by
\begin{equation}
RCM(P_0, \ldots , P_h; w_0, \ldots , w_h) = \argmin{P\in M} \sum_{i=0}^h w_i d(P,P_i)^2
\label{eq:RCM}
\end{equation}
where $d(\cdot,\cdot)$ is the geodesic distance on $M$.
If $M$ is a Euclidean space, then the solution to 
Eq.~\ref{eq:RCM} is the usual affine average of Eq.~\ref{eq:affine}. 

The RCM requires that Eq.~\ref{eq:RCM} has a unique minimizer.
\cite{Karcher:1977ei} provides a condition of existence and uniqueness of 
the solution, which requires all points $P_i$ to be contained inside a 
strongly convex ball, whose maximum radius depends on the curvature of $M$.
In the following, we will refer to this condition as the \emph{Karcher condition}.
If such condition is satisfied, then the RCM is $C^{\infty}$ in both the 
$P_i$'s and the $w_i$'s \cite{afsari2009:me}. 
Unfortunately, the Karcher condition restricts the applicability of the RCM 
to relatively small neighborhoods in the general case.

For any two points $P,Q\in M$, which are connected with a \emph{unique} 
shortest path $\gamma_{P,Q}$, with $\gamma(0)=P$ and $\gamma(1)=Q$, 
it is easy to show that their RCM with weights $(1-w)$ and $w$, respectively, 
is always defined and lies at $\gamma_{P,Q}(w)$. 
This means that weighted averages between pairs of points are defined and smooth, 
as long as such points stay away from each other's 
cut locus\footnote{For a definition of \emph{cut locus, normal ball, convex ball}, and other
terms related to the geodesic metric refer to Appendix \ref{app:geodesic}.}. 
We extend this binary average to the cut locus, too, by picking one arbitrary, 
but deterministically selected, 
shortest path connecting $P$ to $Q$, hence giving up continuity at the cut loci. 
We thus define the \emph{manifold average between two points}
\begin{equation}
\label{eq:average}
\mathcal{A} : M\times M\times [0,1] \longrightarrow M; \ \  (P,Q;w) \mapsto \gamma_{P,Q}(w)
\end{equation}
where $\gamma_{P,Q}$ is a shortest geodesic path joining $P$ to $Q$. 
We have that $\mathcal{A}(P,Q;w)=RCM(P,Q;(1-w),w)$ as long as $P$ and $Q$ 
do not lie on each other's cut locus.
The averaging operator of Eq.~\ref{eq:average} provides the analogous of 
Eq.~\ref{eq:affine2} in the manifold setting.

\begin{figure}[tb]
  \centering
  \includegraphics[width=\columnwidth]{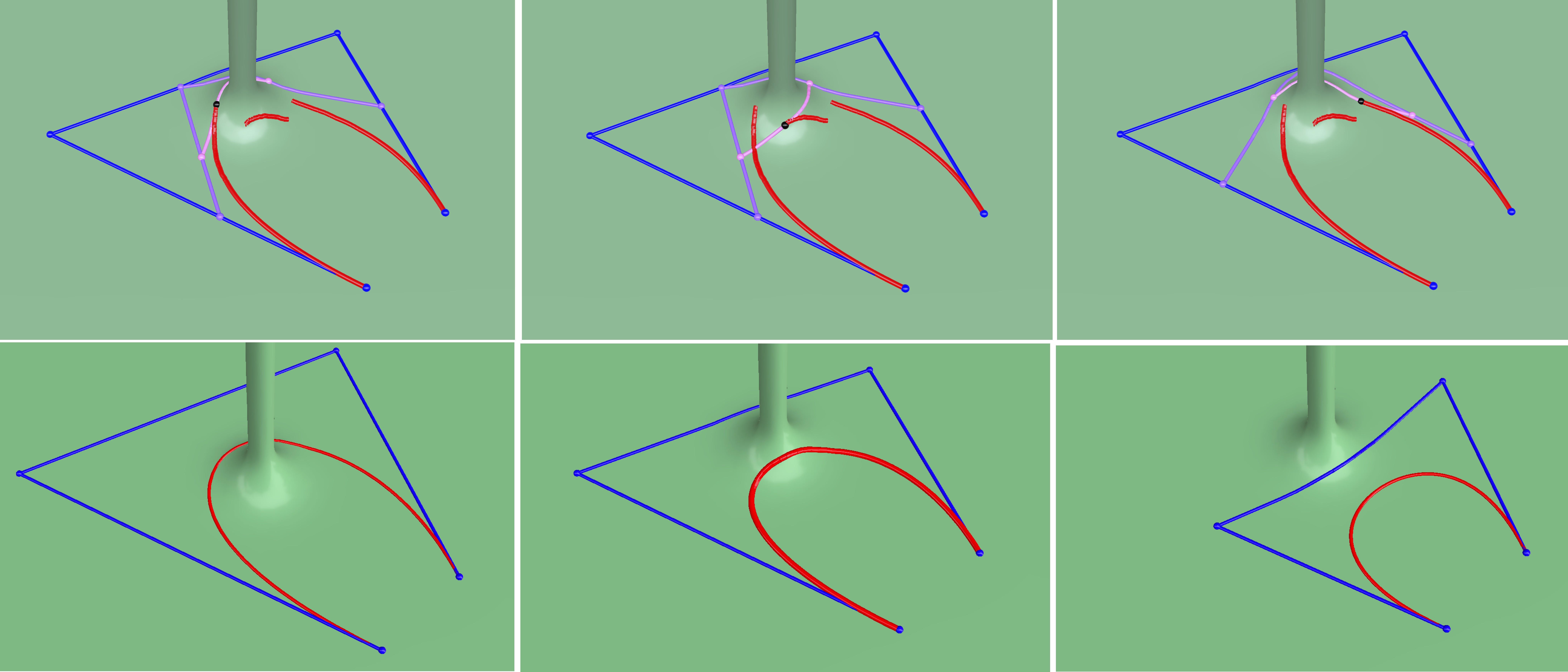}
  \caption{Top: example of a failure case of direct De Casteljau evaluation. 
  The black bullets at the discontinuities correspond to consecutive parameter 
  values near a critical value, and the blue/purple/pink lines provide the 
  De Casteljau construction. Note how the pink line jumps from one side of the 
  pole to the other as $t$ passes critical values, causing discontinuities.
  Bottom: our method always produces a smooth curve regardless of the positioning 
  of the control points. The same control polygon of the top figure generates the curve in the center;
  dragging the handles we may force the curve to pass behind the pole (left) or further shrink (right). 
  Note how the control polygon to the right also switches to the front of the pole, 
  while leaving the smoothness of the curve unaffected.
  }
  \label{fig:DC-fail}
  \Description{Image}
\end{figure}

\subsection{De Casteljau point evaluation}
\label{sub:decasteljau-manifold}
The De Casteljau construction provides a recursive definition, which evaluates 
a Bézier curve at each $t\in[0,1]$ as $\mathbf{b}^k(t) = \mathbf{b}_0^k(t)$, where
\begin{equation}
\label{eq:decasteljau}
\begin{array}{l}
\mathbf{b}_i^0(t)=P_i\\
\mathbf{b}_i^r(t) = (1-t)\mathbf{b}_i^{r-1}(t) + t \mathbf{b}_{i+1}^{r-1}(t)
\end{array}
\end{equation}
for $r=1,\ldots, k$ and $i=0,\ldots, k-r$.
The construction for $k=3$ and $t=0.5$ is exemplified, in the manifold setting, in Fig.~\ref{fig:RDC-OLR} (RDC).

This construction can be extended to the manifold setting in a straightforward 
way, by substituting the affine averages between pairs of points with 
the manifold average $\mathcal{A}$ defined above.
This extension was proposed first by \cite{Park:1995hu}. 
As shown by \cite{Popiel:2007bt}, if all consecutive pairs of control points 
of the control polygon $\Pi$ lie in a totally normal ball, 
then the resulting curve is $C^{\infty}$. 
However, if the constraint is violated, the resulting curve can be discontinuous. 
In fact, even if all shortest geodesics paths in $\Pi$ are unique, 
some pairs of intermediate points involved in the construction may lie on 
each other's cut locus, for some value of the parameter $t$. 
As $t$ passes such critical value, the manifold average $\mathcal{A}$ returns 
a discontinuous result, thus causing a discontinuity in the curve.
Fig.~\ref{fig:DC-fail}(top) illustrates the construction near failure points;
Fig.~\ref{fig:fail}(a) provides another example of failure. 

\begin{figure}[tb]
  \centering
  \includegraphics[width= 0.8\columnwidth]{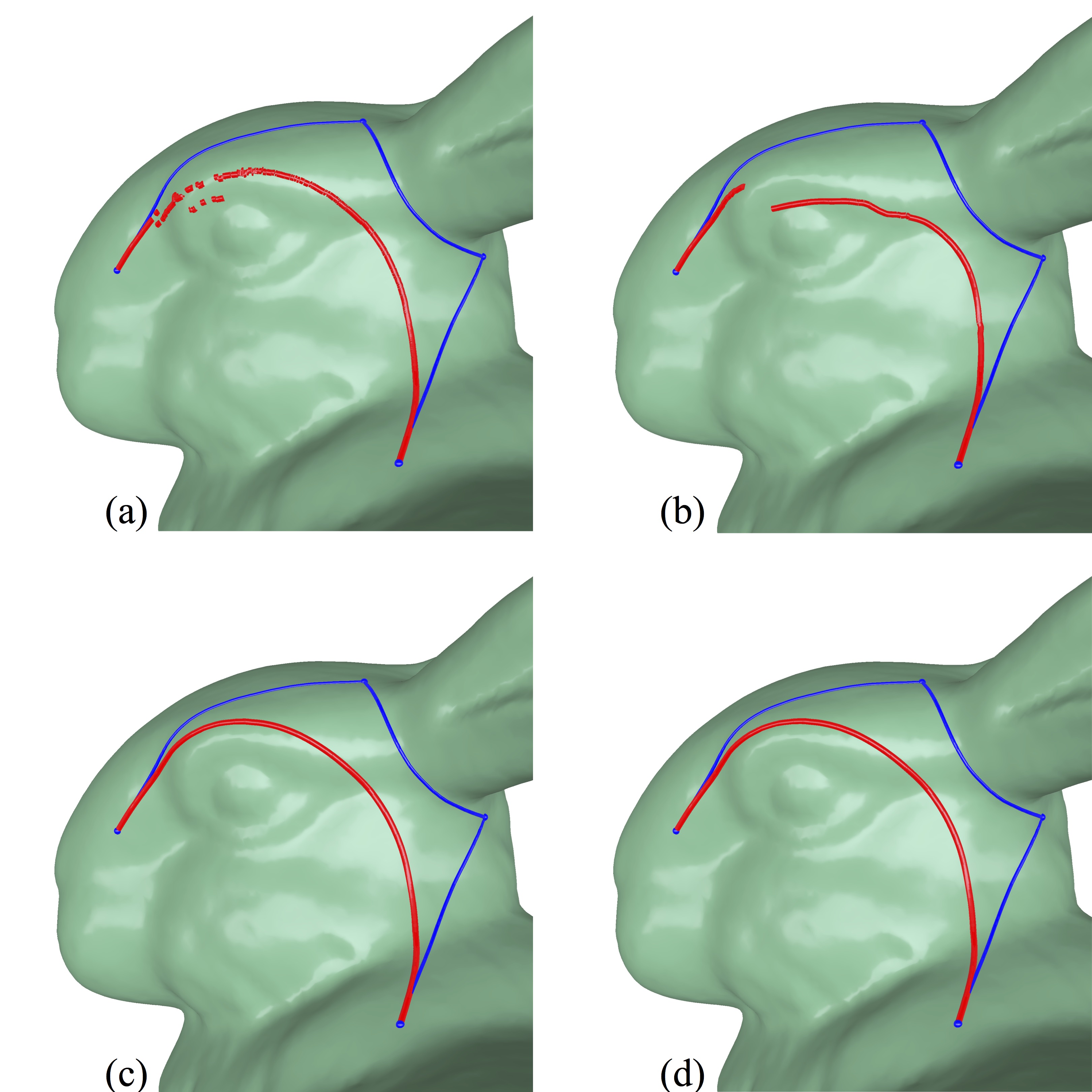}
  \caption{An example of failure in tracing a curve with the direct 
  De Casteljau  (a) and the RCM evaluation (b). 
  The same control polygon gives two smooth and nearly identical curves with 
  the Recursive De Casteljau (c) and the Open-uniform Lane-Riesenfeld schemes 
  (d) described in Sections \ref{sub:DC-manifold} and \ref{sub:LR-manifold}, 
  respectively.}
  \label{fig:fail}
  \Description{Image}
\end{figure}

In principle, this issue could be overcome by applying repeated degree elevation.
In the Euclidean setting, a B\'ezier curve of degree $k$ from polygon $\Pi_k$ 
can be rewritten as a curve of degree $k+1$ on a control polygon $\Pi_{k+1}$. 
Degree elevation is based on weighted averages between pairs 
of consecutive points of $\Pi_{k}$, hence extended to the manifold setting 
by means of operator $\mathcal{A}$.
It can be easily shown that, for any given polygon $\Pi_{k}$, there exist $n\geq k$ such that 
the control polygon $\Pi_{n}$ obtained from $\Pi_{k}$ by repeated application of degree elevation
has all consecutive pairs of control points inside totally normal balls, thus yielding a $C^{\infty}$ curve.
Unfortunately, this approach is prohibitive for interactive usage:
the value of $n$ is hard to estimate; 
the repeated application of degree elevation has 
a cost $O(n^2)$ to generate a polygon with $n$ points;
and point evaluation for a relatively large value of $n$ becomes expensive, too.

\subsection{Bernstein point evaluation with the RCM}
\label{sub:karcher}
A Bézier curve can be evaluated in closed form as an affine sum of all its 
control points:  
\begin{equation}
\label{eq:bernstein}
\mathbf{b}^k(t) = \sum_{i=0}^k B_i^k(t) P_i 
\end{equation}
where the $B_i^k(t)$ are the Bernstein basis polynomials of degree $k$
\[B_i^k(t) = \binom{k}{i} t^i(1-t)^{n-i}.\]
This expression can be rewritten for the manifold case as
\begin{equation}
\label{eq:bernstein-manifold}
\mathfrak{b}^k(t) = 
RCM(P_0,\ldots,P_k; B_0^k(t),\ldots,B_k^k(t)).
\end{equation}
This construction was addressed in \cite{Panozzo:2013eh}, where an approximation of the RCM is proposed,
which is based on an embedding in a higher dimension and Phong projection (see Sec.~\ref{sub:flipOut} for further details).

If the control points are close enough to fulfill the Karcher condition, 
then the resulting curve is $C^{\infty}$, since both the RCM and 
the Bernstein polynomials are $C^{\infty}$. 
However, the Karcher condition is even more restrictive than the constraints 
required for the De Casteljau construction, and it is not likely to be verified 
when the control points lie far apart on a general surface.
If the Karcher condition is not fulfilled, Eq.~\ref{eq:RCM} is no longer convex, 
and it might even have infinitely many minima.
In this case, the curve may be undetermined at some intervals.
Fig.~\ref{fig:fail}(b) provides an example of failure, 
where the RCM has been computed directly by gradient descent.
Similar examples of failure for the approximation of  \cite{Panozzo:2013eh} are discussed in 
Section \ref{sub:flipOut} and shown in Figures~\ref{fig:WA} and \ref{fig:WA-cylinder}.

\subsection{Recursive De Casteljau bisection (RDC)}
\label{sub:DC-manifold}

One step of the De Casteljau evaluation subdivides polygon $\Pi$ into two 
control polygons $\Pi_L$ and $\Pi_R$. See Fig.~\ref{fig:RDC-OLR} (RDC).
The junction point of $\Pi_L$ and $\Pi_R$ lies on the curve. 
The recursive application of this procedure for $t=1/2$ defines a sequence 
of subdivision polygons $\Pi_{DC}^n$, which converges to the Bézier curve.

The extension to the manifold setting is straightforward, by means of the 
point evaluation procedure described in Section \ref{sub:decasteljau-manifold}. 
In the following, we denote this scheme \emph{RDC} for short.
This extension was studied first by \cite{Noakes:1998bo}, and implementations 
were proposed in \cite{Morera:2008jk,geometrycentral}.

Concerning the convergence and smoothness of the limit curve, we have the following result: 

\begin{prop}
\label{sub-smoothness-RDC}
For any given control polygon $\Pi_{k}$, with $k>1$, the RDC subdivision converges to a limit curve that is $C^1$ continuous.
\end{prop}

In Appendix \ref{app:proofs} we provide a proof for $k=2,3$ and a sketch of proof for generic $k>3$.  
Note that \cite{Morera:2008jk} gave a proof of smoothness just for the special case of curves 
embedded on a triangle mesh, showing that the limit curve is $C^1$ once 
the strip of triangles it intersects is flattened on a plane. 
Our proof is given in general for a smooth manifold. 
It remains an open question whether the RDC scheme produces curves with higher smoothness. 

%
%

\begin{figure}[tb]
  \centering
  \includegraphics[width=\columnwidth]{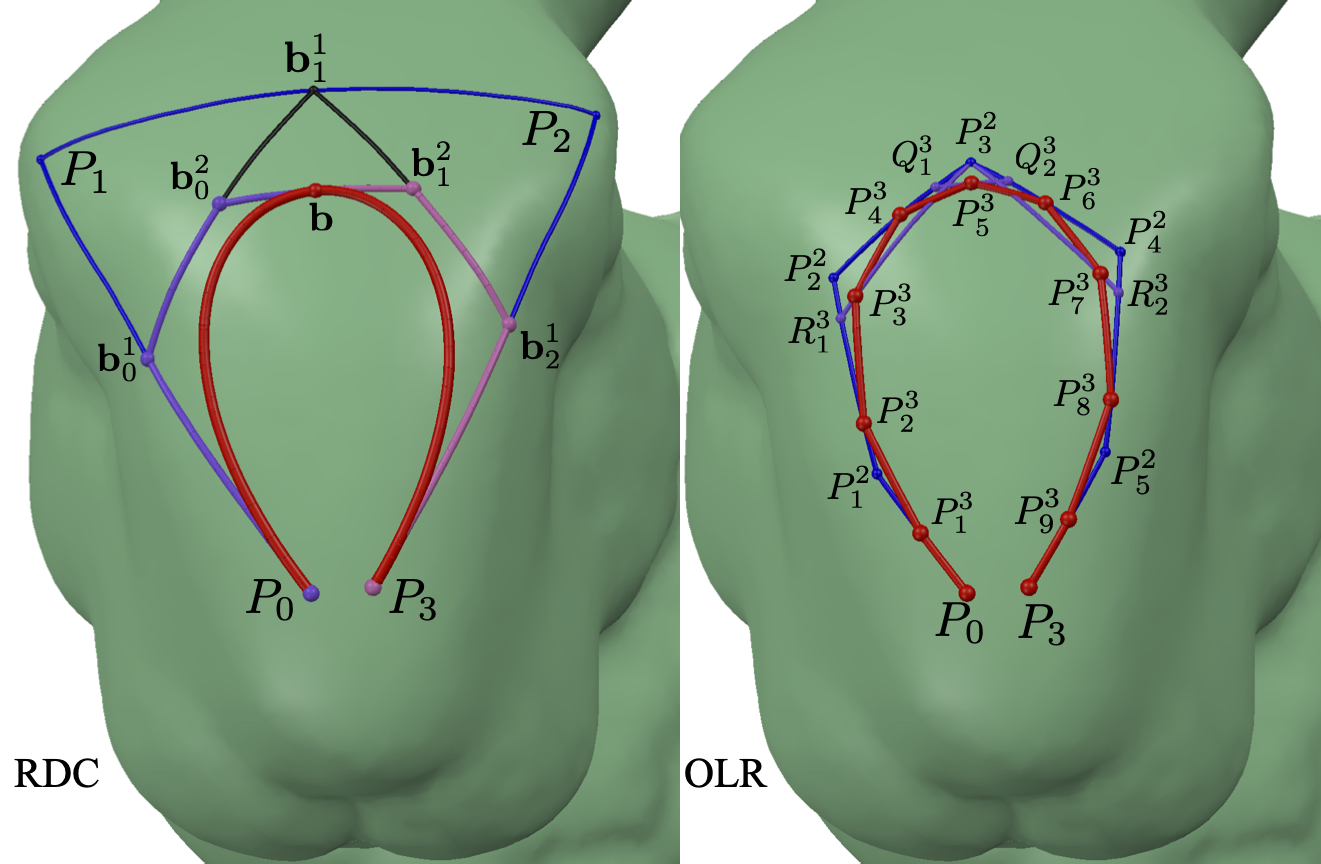}
  \caption{The constructions at the basis of the RDC and ORL schemes for the 
  same control polygon for the cubic case ($k=3$). 
  \emph{RDC (left)}: The control polygon (blue) is split into a chain of two 
  control polygons (purple and pink) by computing three shortest geodesic paths. 
  The limit curve is depicted in red. Here we show only the first subdivision.
  \emph{OLR (right)}: One step of subdivision from $\Pi^2$ (blue) to 
  $\Pi^3$ (red polygon). The even points $P_{2j}^3$, as well as the intermediate 
  points $Q_i^{3}$ and $R_i^{3}$ lie on segments of $\Pi^{2}$ and are evaluated first.
  The evaluation of each odd point $P_{2j+1}^{3}$ requires computing one shortest 
  geodesic path (purple).
  This construction corresponds to one midpoint subdivision followed by two steps 
  of smoothing by averaging consecutive points.
   }
  \label{fig:RDC-OLR}
  \Description{Image}
\end{figure}

\subsection{Open-uniform Lane-Riesenfeld Subdivision (OLR)}
\label{sub:LR-manifold}

To achieve higher continuity, we propose another subdivision scheme,
which is novel in the manifold setting. 

In the Euclidean setting, a Bézier curve of degree $k$ can be represented with an open-uniform 
B-spline\footnote{A B-spline is said to be open-uniform if it is uniform, 
except at its endpoints, where repeated knots are inserted to make the curve 
interpolate the endpoints of its control polygon.} of degree $k$ (order $k+1$), 
having the same control polygon $\Pi_k$, and knot vector $(00\ldots011\ldots1)$, 
where the $0$ and $1$ are repeated $k+1$ times.
Repeated knot insertion at the midpoint of all non-zero intervals in the knot 
vector produces a sequence of open uniform B-splines, 
all describing the same curve, whose control polygons $\Pi_{LR}^n$ converge 
to the curve itself \cite{Cashman:2007jh}.
This subdivision process follows an open-uniform Lane-Riesenfeld scheme (OLR). 
The control points are nearly doubled at each level of subdivision, by applying 
the standard even-odd stencils of the uniform Lane-Riesenfeld scheme 
``in the middle'' \cite{Lane:dq}, while stencils for end conditions are applied 
near the endpoints. 
A number of $2k-2$ special stencils are needed at each end of the polygon.  
The full constructions for the quadratic and cubic cases, as well as a sketch 
of the construction for a generic degree $k>3$, are provided in Appendix \ref{app:subLR}.
One step of subdivision for $k=3$ and $n=3$ is exemplified, in the manifold setting, 
in Figure~\ref{fig:RDC-OLR} (OLR).
Note that $n=3$ is the first level in which the stencils of the uniform LR subdivision apply. 

It is important to notice that, all affine averages necessary to compute 
the stencils in this scheme can be factorized into weighted averages between 
pairs of points, as shown in Appendix \ref{app:subLR}. 
This feature is intrinsic to the uniform LR scheme, and it is easily 
generalized to the end conditions in the open-uniform scheme. 
Therefore, we extend this scheme to the manifold setting, by substituting 
each affine average with the corresponding application of the manifold average 
$\mathcal{A}$. We omit the details for the sake of brevity. 

Concerning the convergence and smoothness of the resulting curve, we have the following result 
(proof in Appendix \ref{app:proofs}):

\begin{prop}
\label{sub-smoothness-OLR}
For any given control polygon $\Pi_{k}$, the OLR subdivision converges to a limit curve that is $C^{k-1}$ continuous, possibly except at its endpoints. 
At the endpoints, the limit curve interpolates the endpoints of polygon $\Pi_{k}$ and it is tangent to it.
\end{prop}

It follows from the proposition above that single cubic B\'ezier curves in the manifold setting warrant $C^2$ continuity, while different segments can be joined to form splines with  $C^1$ continuity. 
It remains an open problem how to build splines with $C^2$ continuity at junction points \cite{Popiel:2007bt}. 
Note that the construction of interpolating splines with $C^2$ continuity is all but simple even in the Euclidean setting \cite{Yuksel2020}.

\begin{figure}[tb]
  \centering
  \includegraphics[width=\columnwidth]{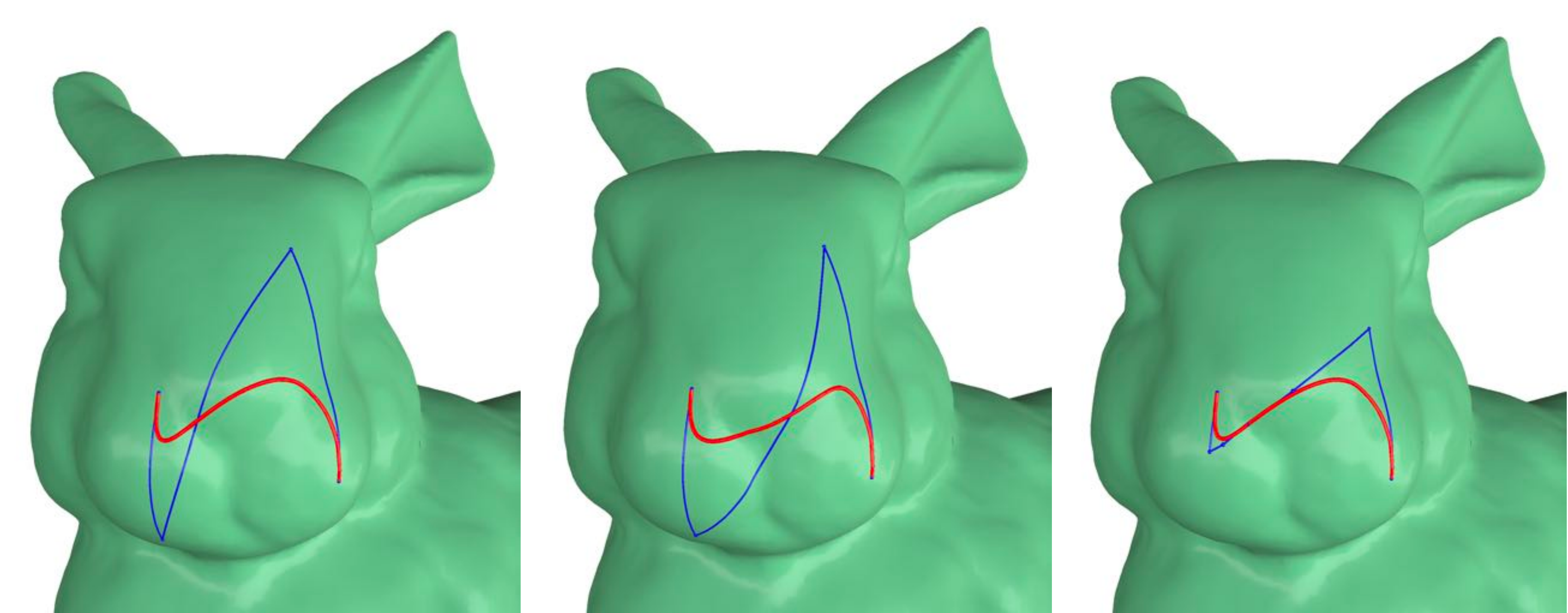}
  \caption{Left, center: the geodesic line corresponding to the central segment of the control 
  polygon can take two different routes at the cut locus of its endpoints, 
  thus producing two different curves; 
  in practice, the curve will jump between the two configurations when dragging a 
  control point across a cut locus. 
  Right: splitting the curve by point insertion makes the selected configuration stable upon dragging.
 }
  \label{fig:jump}
  \Description{Image}
\end{figure}

\subsection{Discussion}
\label{sub:discussion-manifold}

The first main contribution of our work is the above analysis, from which
we can derive safe schemes for lifting Bézier curves to the manifold domain.

The Bernstein point evaluation cannot be used safely on manifolds, since the 
Riemannian center of mass is fragile and can be used only ``in the small''. 
Using direct De Casteljau evaluation is also problematic.
In fact, while the manifold average $\mathcal{A}$ for a pair of points is 
defined over the whole domain, and smooth everywhere except at the cut loci,
the discontinuity of operator $\mathcal{A}$ at the cut loci makes the direct 
De Casteljau evaluation fragile. 

Conversely, subdivision schemes can be safely defined with repeated averages 
based on operator $\mathcal{A}$ and can work on any control polygon, since 
the arbitrary choices made with operator $\mathcal{A}$ at the cut loci 
do not affect convergence and smoothness:
%
%
the RDC scheme, based on the recursive De Casteljau bisection, guarantees 
$C^1$ continuity; 
while the OLR scheme, based on open-uniform Lane-Riesenfeld subdivision, 
guarantees $C^{k-1}$ smoothness for a curve of degree $k$.
Both schemes can be implemented easily and efficiently 
via repeated geodesic averages.

\subsubsection{Limitations}
\label{sub:limitations-manifold}

In principle, the arbitrary choices made with operator $\mathcal{A}$ 
at the cut loci can lead to different curves for a given control polygon.
In practice, since all our algorithms are deterministic, the same paths 
will always be chosen at the cut loci, thus returning the same curve.
However, the curve may jump to a different configuration for 
small displacements of control points, which make some of the paths 
in the construction cross a cut locus.
In Fig.~\ref{fig:jump} (left, center), a tiny displacement of one control point takes 
one of the shortest paths in the control polygon to a drastically different 
route, resulting in a different curve;
see the bottom part of Fig.~\ref{fig:DC-fail} for another example. 
Note that jumps occur quite rarely, as the cut locus of each point covers 
a set of zero measure on the manifold.
This fact is intrinsic to the discontinuity of the manifold metrics and 
constitutes an essential limitation to the design of splines in the manifold setting,
independently of the approach adopted.

This limitation can be circumvented easily, by means of splines containing more 
control points, instead of single B\'ezier segments.
See Fig.~\ref{fig:jump} (right).  
This can be done easily with point insertion that can be used to constrain the 
curve to a desired path, as is customarily done in curve design, and motivates 
the algorithms we present in Sections \ref{sub:DC-insert} and \ref{sub:insertion}. 

Automatic solutions would also be possible. 
It is easy to check when a curve ``jumps'' while dragging a control point; 
in that case, the control polygon may be split, e.g., by adding the midpoint 
of the curve before displacement as a new control point. 
Another approach would be to homotopically 
deform the lines of the control polygon while dragging a control point. 
The point-to-point geodesic algorithm that we present in Sec.~\ref{sub:geodesicprimitives}
can support such task in a straightforward manner. 
However, the latter solution could generate a control polygon consisting 
of lines that are arbitrarily far from being shortest geodesics, 
thus hindering all the theory about convergence and smoothness of the result; even worse, 
the curve would not be described by its control points only, but it would depend on its construction, too. 
In our user interface, we decided to avoid using automatic methods to warrant maximum flexibility to the user.

\section{Practical Algorithms}
\label{sec:algorithm}

We now focus deriving practical algorithms for the RDC and OLR schemes.
We provide algorithms for approximating the curve with a geodesic polyline 
(curve tracing), evaluating a point on the curve for a given parameter value 
(point evaluation), and splitting a curve at a given point into a spline with two segments (point insertion).
The algorithm for curve tracing with the RDC scheme is equivalent 
to the one proposed in \cite{Morera:2008jk} and it is described briefly for completeness; 
the other five algorithms are novel.

We support surfaces represented as triangle meshes, and target interactivity for
long curves and meshes of millions of triangles. 

In order to develop our algorithms, we assume to have procedures for 
(1) computing the point-to-point shortest path between pairs of points of $M$; 
(2) evaluating a point on a geodesic path at a given parameter value; and 
(3) casting a geodesic path from a point in a given direction.
In all these cases, we consider generic points on the surface, not just
the vertices of the mesh.
The computational details of such procedures, as well as additional algorithms 
to support interactive control, are provided in Section \ref{sec:implementation}.

\subsection{Algorithms for the RDC scheme}
\label{sub:DC-algorithm}

\subsubsection{Curve tracing.}
We trace Bézier curves on surfaces by approximating them with a geodesic polygon.
The tracing algorithm is a recursive subdivision that, at each step, takes a 
geodesic polygon $\Pi$ and produces two sub-polygons $\Pi_L$ and $\Pi_R$. 
Recursion is initialized by computing the $k$ shortest paths that constitute 
the polygon connecting the initial control points $P_0,\ldots,P_k$.

To split a control polygon, we compute a sequence of geodesic polygons 
$\Pi^i$ for $i=0,\ldots,k$, each containing $k-i$ segments, 
where $\Pi^0=\Pi$, and $\Pi^k$ degenerates to the midpoint of the curve. 
Since $\Pi^0$ is known from recursion, this requires computing a total of $k(k-1)/2$ 
further geodesic paths, each joining the midpoints of the segments in $\Pi^i$, 
in order to produce the points of $\Pi^{i+1}$.
The polygon  $\Pi_L$ is built by collecting all sub-paths that connect the 
starting points of the polygons in the sequence, namely $\Pi^i[0]$ to 
$\Pi^{i+1}[0]$ for $i=0,\ldots,k-1$. The polygon $\Pi_R$ is built likewise, 
by using sub-paths connecting the end points of subsequent polygons. 

We support both uniform and adaptive subdivision. For uniform subdivision, 
a maximum level of recursion is chosen either by the user, 
or automatically computed on the basis of the total length $L(\Pi)$ of 
the initial polygon $\Pi$, and a threshold $\delta>0$. Since the paths forming 
the subdivided polygon are shrinking through recursion, 
then after $\lceil\log_2(L(\Pi)/\delta)\rceil$ recursion levels, the length of a 
geodesic path in the output will be no longer than $\delta$.
For adaptive subdivision, we stop recursion as soon as the angles between 
tangents of consecutive segments of $\Pi$ differ for less than a given 
threshold $\theta$. For a small value of $\theta$, this suggests that the curve 
can be approximated with a geodesic polyline connecting 
the points of $\Pi$.

\subsubsection{Point evaluation.}
In point evaluation, we compute the location of the point at a value $\bar{t}$ 
on the curve.
Since we are dealing with a subdivision curve, point evaluation requires 
traversing the recursion tree with a bisection algorithm.
Each time we descend one level, we split the control polygon of that level 
as described above, but only compute the sub-polygon that contains $\bar{t}$.
We stop recursion with the same criteria listed above.

The point at $\bar{t}$ is computed by direct De Casteljau evaluation 
at value $\bar{t}$ on the leaf control polygon.
Here we are assuming that the control polygon in the leaf node is short 
enough to support direct De Casteljau evaluation, and we use it to 
approximate the limit point on the subdivided curve.
By using arguments of proximity, as in \cite{Wallner:2006gg}, it can be shown 
that this approximation converges to the limit curve, as the subdivision polygon 
is subdivided further. 

\begin{figure}[tb]
  \centering
  \includegraphics[width=\columnwidth]{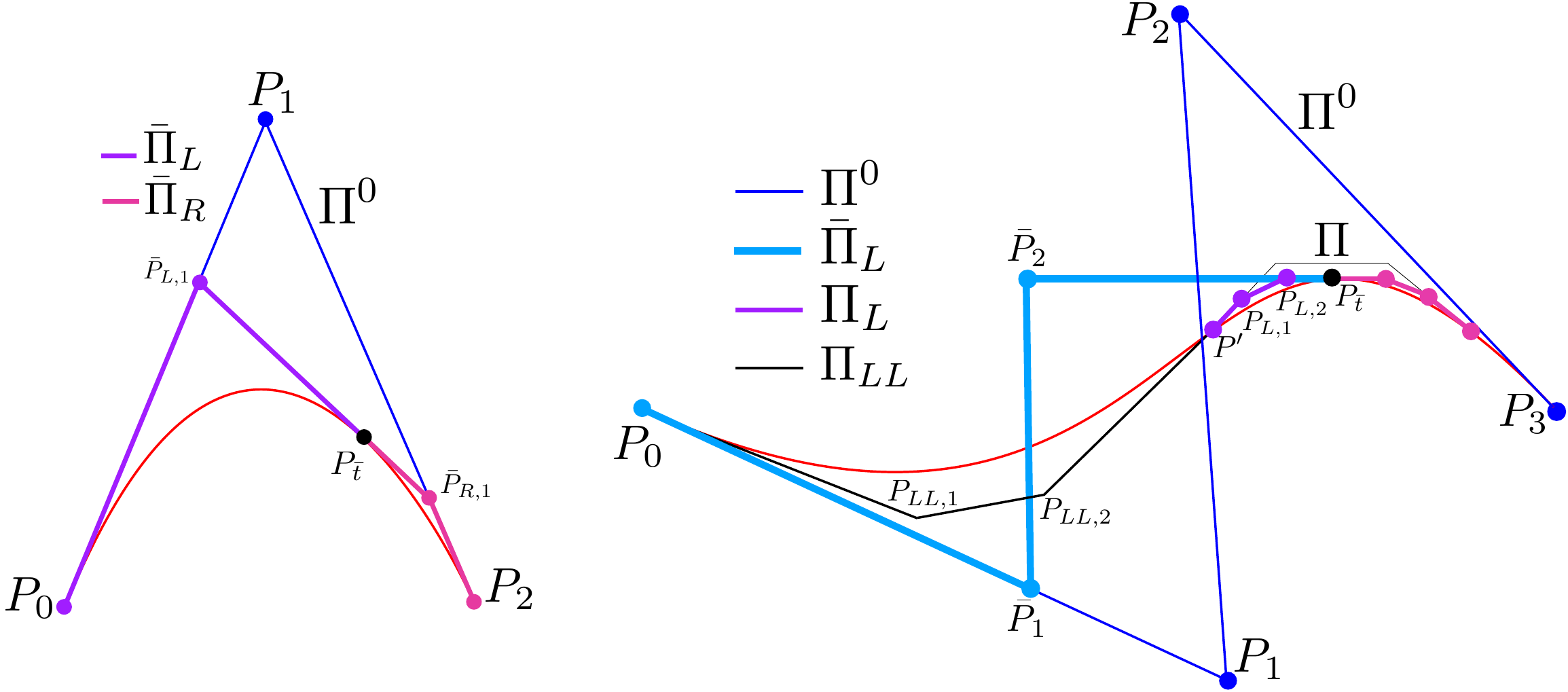}
  \caption{The point insertion algorithm for quadratic, $k=2$, and cubic case, $k=3$
  for the left side of the curve only.
  \emph{Left, $k=2$}: To warrant $\mathcal{C}^1$ continuity at the junction 
  point $P_{\bar{t}}$, the new control points $\bar{\Pi}_L$ and $\bar{\Pi}_R$ 
  are found with the De Casteljau construction, 
  computing $\mathcal{A}(P_i,P_{i+1},\bar{t})$ for $i=0,1$, 
  and repeating the same average on the resulting points.  
  \emph{Right, $k=3$}: The control polygon $\bar{\Pi}_L$ (light blue) 
  encompassing the curves described by two cubic polygons $\Pi_{LL}$(black) 
  and $\Pi_L$ (purple) is built by shortening the first segment of $\Pi^0$ 
  and extending the last segment of $\Pi_L$.
  }
  \label{fig:DC-insert}
\end{figure}

\subsubsection{Point insertion.}
\label{sub:DC-insert}
With point insertion, a user can split a curve at a given point $P_{\bar{t}}$, 
and obtain a spline consisting of two Bézier curves, from $P_0$ to $P_{\bar{t}}$ and 
from $P_{\bar{t}}$ to $P_k$, which coincides with the input curve. 
This is used to add detail during editing.
While this computation is exact in the Euclidean setting, the identity of curves 
before and after point insertion cannot be guaranteed in the manifold setting.
Here we provide a solution that interpolates the endpoints, 
as well as point $P_{\bar{t}}$, and is \emph{nearly} equal to the input curve.  
This is usually sufficient for practical purposes. 
The algorithm for a generic value of $k$, which requires climbing a path in the recursion tree, 
is described in Appendix \ref{app:pointinsertion}.
In the following, we describe the simpler cases of $k=2$ and $k=3$, as illustrated in 
Fig.~\ref{fig:DC-insert}.

In the cubic case, we descend the recursion tree as in the previous algorithm, 
in order to find the leaf containing the splitting point $P_{\bar{t}}$.
Then we split the control polygon $\Pi$ associated to the leaf node at value 
$\bar{t}$, thus obtaining the two polygons $\Pi_L$ and $\Pi_R$, respectively.
We process the two halves independently. Here we show the algorithm for the 
polygon $\bar{\Pi}_L$ defining the curve between $P_0$ and $P_{\bar{t}}$. 
The construction of the other half is symmetric.

Let $\Pi_{LL}$ describe the portion of curve to the left of $\Pi_L$.
let $P'$ be the junction point of  $\Pi_{LL}$ and $\Pi_L$, and let us denote 
\[\begin{array}{l}
\Pi_{LL}=(P_{LL,0}=P_0,P_{LL,1},P_{LL,2},P_{LL,3}=P')\\
\Pi_L=(P_{L,0}=P', P_{L,1}, P_{L,2}, P_{L,3}=P_{\bar{t}})\\
\bar{\Pi}_L=(\bar{P}_0=P_0,\bar{P}_1,\bar{P}_2,\bar{P}_3=P_{L,3})
\end{array}
\]
where we just need to determine the points $\bar{P}_1$ and $\bar{P}_2$ of $\bar{\Pi}_L$.
Note that, in order to preserve the correct tangents at the endpoints of 
$\bar{\Pi}_L$, $\bar{P}_1$ must lie on the geodesic path connecting $P_0$ to $P_1$, 
while $\bar{P}_2$ must lie on the extension of geodesic path connecting $P_{L,3}$ to $P_{L,2}$.
Point $\bar{P}_1$ is found trivially: since $P_{\bar{t}}$ is an endpoint of the 
curve defined by $\bar{\Pi}_L$, and it lies at parameter $\bar{t}$ on the input 
polygon $\Pi$, then $\bar{P}_1$ must be at distance $\bar{t}\cdot d(P_0,P_1)$ from $P_0$. 
In order to find $\bar{P}_2$, we consider the concatenation of polygons 
$\Pi_{LL}$ and $\Pi_L$, which provides the construction to evaluate point $P'$ 
from $\bar{\Pi}_L$.
Let $t'$ be the parameter corresponding to $P'$ on the input curve, 
then we have
$$d(P_{L,3},P_{L.2})=(1-\frac{t'}{\bar{t}})\cdot d(P_{L,3},\bar{P}_2).$$
Therefore, we conclude that $\bar{P}_2$ is obtained by extending the geodesic 
line from $P_{L,3}$ to $P_{L,2}$ for a length $d(P_{L,2},P_{L,3})\cdot \frac{t'}{(\bar{t}-t')}$.

In the quadratic case, we have to compute two polygons 
$\bar{\Pi}_L = (P_0,\bar{P}_{L,1},P_{\bar{t}})$ 
and $\bar{\Pi}_R = (P_{\bar{t}},\bar{P}_{R,1},P_2)$, which must concatenate 
with smoothness at least $C^1$. Following the same strategy used in the previous 
case is not robust, as it could hinder the smoothness of the curve at $P_{\bar{t}}$.
We rather evaluate $\bar{P}_{L,1}$ on path $P_0,P_1$ and, similarly, 
$\bar{P}_{R,1}$ on path $P_1,P_2$, both at parameter $\bar{t}$ along the 
respective geodesic lines. 
Then we approximate $P_{\bar{t}}$ with the point at parameter $\bar{t}$ along 
the path joining $\bar{P}_{L,1}$ to $\bar{P}_{R,1}$.

\subsection{Algorithms for the OLR scheme}
\label{sub:LR-algorithm}

\subsubsection{Curve tracing}
\label{sub:tracing}

For uniform subdivision, our OLR scheme can be easily expanded up to a 
certain level $\bar{n}$, and the curve approximated with the geodesic 
polygon $\Pi^{\bar{n}}$. The maximum expansion level $\bar{n}$ can be set as 
in the corresponding RDC algorithm. At each level of subdivision, we obtain the 
vertices of the refined polygon by applying the subdivision stencils described 
in Appendix \ref{app:subLR}, where affine averages between pairs of points are 
substituted with the manifold average $\mathcal{A}$. We omit the details for brevity. 
The cubic case is illustrated in Fig.~\ref{fig:RDC-OLR} (OLR).

Notice that the uniform subdivision, as described above, defines a (virtual and infinite) binary tree of intervals, 
that we call the \emph{expansion tree}: the root of the expansion tree corresponds to the whole interval $[0,1]$, 
while a generic node $[t_j^i,t_{j+1}^i]$ at level $i$ is split in the middle into two intervals at level $i+1$. 
The node $[t_j^i,t_{j+1}^i]$ encodes a segment of B-spline, defining the curve in the corresponding interval,
with control points $(P_{j-k}^i,\ldots,P_j^i)$. 
One more level of subdivision splits this interval into two sub-intervals 
$[t_{2j}^{i+1},t_{2j+1}^{i+1}]$ and $[t_{2j+1}^{i+1},t_{2j+2}^{i+1}]$  and generates 
$k+2$ new control points, which depend just on $(P_{j-k}^i,\ldots,P_j^i)$: 
the first $k+1$ points are associated to the interval to the left, and the last 
$k+1$ to the interval to the right, with an overlap of $k$ control points between the two sets.
The expansion tree is defined implicitly and it needs not being encoded.

We exploit the structure of the expansion tree to design an algorithm for adaptive subdivision, which is controlled
by the same stopping criterion used for the RDC scheme, i.e., we stop the expansion of a node as soon as 
the angle between consecutive segments of the polygon is small enough.
The algorithm corresponds to visiting a subtree of the expansion tree in depth-first order; 
a leaf of the subtree is a node of the expansion tree where we stop recursion. 
During the visit, at each internal node, we split the interval as described above, and generate the control points for its two children to continue the expansion; 
while at each leaf, we generate the nodes of the output polygon. 

Depth-first traversal guarantees that leaves are visited left to right: 
the leftmost leaf in the expansion tree is the first one to produce an output, adding all its control points;
all other leaves add just their rightmost control point to the output. 


The final approximation of the curve is obtained by connecting the output points pairwise with shortest geodesic paths.

Note that,  it is not necessary to encode the subtree visited by the algorithm.
It is just sufficient to encode the path in the expansion tree connecting the root to the current node, 
storing at each node its corresponding interval, and its control polygon.


\subsubsection{Point evaluation}
\label{sub:eval}

The point evaluation algorithm is analogous to the one for the RDC scheme, by descending a path in the expansion tree described above.
Given interval $[t_j^i,t_{j+1}^i]$ containing $\bar{t}$ at subdivision level $i$, we only need to compute, by applying the proper stencils, the $k+1$ points 
corresponding to its sub-interval containing $\bar{t}$ at the next level.

Once recursion stops, we assume that all pairs of consecutive control points in the current interval lie in a totally normal ball. 
Here we evaluate the curve directly with a manifold 
version of the de Boor algorithm \cite{Farin2001}, which works on repeated 
averages and can be obtained by substituting the affine averages with the 
manifold average $\mathcal{A}$, just like the direct De Casteljau evaluation.
We omit the details for brevity.

The same remarks we made for the RDC scheme about approximation and convergence 
in the limit apply here, too.

\subsubsection{Point insertion}
\label{sub:insertion}

This algorithm is analogous to the one described for the RDC scheme.
We descend the recursion tree as in the point evaluation algorithm.
When reaching the leaf containing the splitting point $P_{\bar{t}}$, 
we convert the control polygon of the uniform B-spline in that leaf into the 
corresponding control polygon of the Bézier curve, by applying the standard conversion 
formula reported in Appendix \ref{app:BtoBezier}, where affine averages are substituted 
with the manifold average $\mathcal{A}$.
Then we proceed as described for the RDC scheme.


\section{Implementation and User Interface}
\label{sec:implementation}

The implementation of the algorithms described in the previous section 
rests on a few geodesic operations that we describe in this section together
with operations required to support the user interface. All operations are 
implemented in C++ and released as open source in \cite{anonymous}.

\subsection{Data structures}
\label{sub:datastructure}

We encode a  triangle mesh $M$ with a simple indexed data structure 
consisting of three arrays encoding the vertices, the triangles, and triangles adjacencies,
which also provide the dual graph having the triangles as nodes.

We need to deal with generic points lying on the mesh, not just its vertices.
A \emph{mesh point} $P$ is encoded as a triple $(t,\alpha,\beta)$ where $t$ is 
the triangle index and $\alpha,\beta$ are the barycentric coordinates of $P$ in $t$.
A vertex $v$ of $M$ can also be encoded as a generic mesh point, 
by means of any of its incident triangles.  

A geodesic path connecting two mesh points $P$ and $Q$ is encoded with a 
triangle strip $(t_0,\ldots,t_h)$, where $t_0$ and $t_h$ contain $P$ and $Q$, 
respectively, and an array of real values $l_1,\ldots,l_{h-1}$, where $l_i$ encodes 
the intercept of the path with the edge $e_i$ common to $t_i, t_{i+1}$, 
parametrized along $e_i$.

\subsection{Basic geodesic primitives}
\label{sub:geodesicprimitives}

\paragraph{Point-to-point shortest path.} 
The literature offers several techniques for computing shortest paths \cite{crane2020survey}.
We propose an algorithm to compute locally shortest geodesic paths, 
which is derived by combining insights from the works of \cite{Xin:2007,Lee:1984hn}. 
The algorithm consists of three phases: (i) extraction of an initial strip; (ii) shortest path in a strip; and (iii) strip  straightening.

Phase (i), which has been overlooked in several previous works, is critical as it can become the bottleneck 
on large meshes (see, e.g., the discussion in \cite{sharp2020flipout} 5.2.1). 
Given two mesh points $P$ and $Q$, we 
compute a strip of triangles that connects them, performing a search on the dual graph. 
We experienced a relevant speedup over the classical Dijkstra search by using a shortest path algorithm based on the SLF and LLL heuristics \cite{Bertsekas:1998}, which do not require a priority queue, but just a double ended queue. 
The SLF and LLL heuristics govern the insertion and extraction of weighted nodes in the queue.
We weight each node as in a classical A* search, with the sum of its current distance from the source plus its Euclidean 3D distance to the target.  
This heuristic prioritizes the exploration of triangles closer to the destination in terms of Euclidean distance, improving performance in most models.

In phase (ii), the strip is unfolded in the 2D plane and the shortest path within it is
computed  in linear time with the funnel algorithm \cite{Lee:1984hn}.
See Fig.~\ref{fig:funnel}(a-b-c) for an example.

In phase (iii), in order to obtain the locally shortest path on the mesh, we remove reflex vertices from the strip where possible. 
To this aim, \cite{Xin:2007} finds the reflex vertices that can be removed by computing angles about a vertex inside and outside the strip, respectively. 
However, in our experiments, the computation of angles slows down the algorithm, because the star of each reflex vertex 
is retrieved from a data structure that is not in cache memory. 
Instead, we select the reflex vertex $v$ that creates the largest turn in the polyline and, similarly to \cite{Xin:2007},
we update the strip by substituting the current semi-star of $v$ inside the strip with its other semi-star. 
We perform the unfolding and the funnel algorithm on the new strip: if $v$ still remains on the path, then it is frozen; we repeat this procedure until all reflex vertices either are removed or become frozen.
See Fig.~\ref{fig:funnel}(d-e-f) for an example. 


\begin{figure}[t!]
  \centering
  \includegraphics[width=\columnwidth]{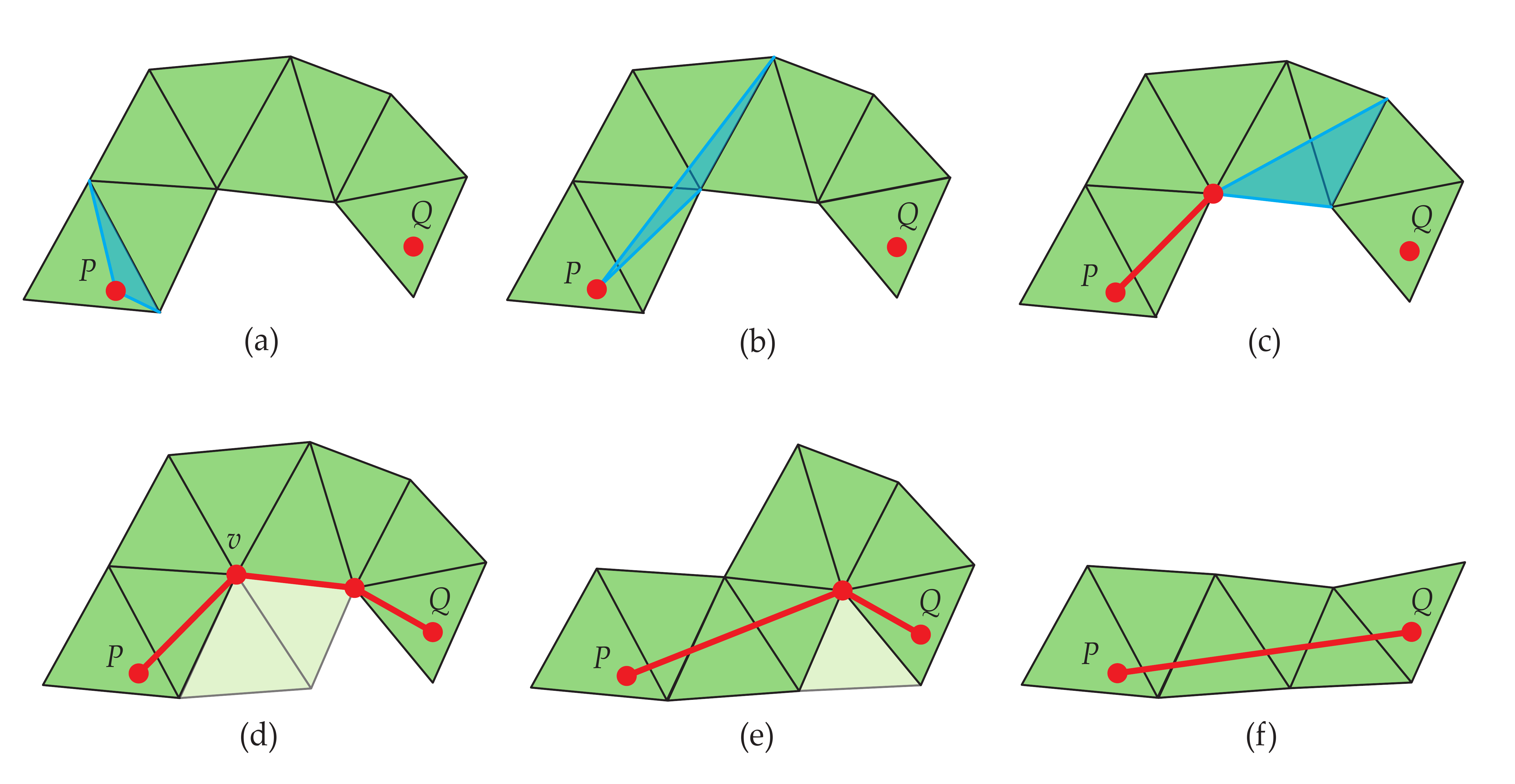}
  \caption{Shortest path computation. 
  Given a source $P$ and a target $Q$ an initial strip of triangles 
  connecting them is found with a search on the dual graph of the mesh. 
  (a) A shortest path within the strip is found by propagating a funnel, 
  which is initialized with its apex at $P$ and its front at the first edge 
  crossing the strip.
  (b) The edges of the strip are processed one by one, to tighten the front 
  of the funnel. 
  (c) When the funnel collapses, a new vertex, called a pseudo-source, 
  is added to the path and the apex of the funnel is moved to the 
  pseudo-source.
  (d) When $Q$ is reached, some reflex vertices may still lie on the path.
  (e) Reflex vertices are analyzed for possible removal, starting at the 
  vertex $v$ causing the sharpest turn.
  (f) The final path is found when no more reflex vertices can be removed.
  }
  \label{fig:funnel}
  \Description{Image}
\end{figure}

\paragraph{Straightest geodesics.}
A straightest geodesic is traced starting at a mesh point $P$ and following one 
given direction $\mathbf{u}$ in the tangent plane $T_P M$ for a given length.
This is done by unfolding the triangles that are crossed by the line as it is being traced
and intersecting their edges with the line in the 2D domain.
The ending results is a straight 2D line that crosses a strip of unfolded 2D triangles,
which can be mapped back into a geodesic surface path thanks to our representation. 
%
In case the path intersects a vertex $v$, then we follow \cite{polthier1998straightest}, 
reflecting the incoming direction about $v$ in its tangent plane. 

\paragraph{Parallel transport.}
We use an approach similar to \cite{Knoppel:2013bd}. 
Each triangle has its own tangent 2D frame of reference.
A direction at a mesh point is represented by 2D coordinates with respect
to the frame of its containing triangle.
When a direction must be parallely transported from a mesh point $P$ to another $Q$,
a rotation must be applied to take into account the fact that $P$ and $Q$ may lie
on different triangles, with different frames of reference. 
To do so, the geodesic path from $P$ to $Q$ is computed, then the strip containing the
path is unfolded so that the coordinates of the axis of the frame of reference of $Q$'s triangle
can be expressed with respect to the frame of reference of $P$'s triangle, hence the rotation
between the two frames can be obtained.


\subsection{User interface}
\label{sub:user}
Leveraging the proposed algorithms, we developed a graphical application to 
allow users to interactively edit splines on meshes, imitating the same 
interaction of established 2D vector graphics tools. We focus on cubic curves
since they are the most used in 2D.
Our application supports the editing of curves by moving, adding, 
and deleting control points, and by translating, scaling and rotating 
whole splines on the surface domain.
All these operations are supported by using the geodesic primitives 
just described.
Here we describe the main editing feature, referring the reader to the 
supplemental video for a demonstration.

\paragraph{Curve editing.}
Borrowing the editing semantic from 2D tools, control points are 
distinguished in \emph{anchor points} and \emph{handle points}. 
Anchors are those points where two Bézier curves are joined, hence a spline 
passes through all of its anchor points.
The preceding and following control points of an anchor are its associated 
handle points.
The handle points of an anchor determine two segments, both starting at the 
anchor itself. 
A spline is tangent to those segments at the anchor points.

In the 2D setting, when an anchor is dragged, the two tangent segments move 
with it and so do the associated handle points.
To obtain the same behavior on the surface, when moving an anchor point 
from $P$ to $P'$, we find the two tangent directions of the tangent segments at $P$.
Then, for each such segment, we trace a straightest geodesics starting at 
$P'$ and for the same length of the segment, in the direction of its tangent, 
rotated by the parallel transport from $P$ to $P'$.
The endpoint of each segment is the new position of the corresponding handle point.

In the 2D setting the user can impose an anchor to be "smooth", i.e. the two 
associated tangent segments are always colinear, which automatically 
ensure $C^1$ continuity at the anchor point.
To provide the same functionality on the surface, whenever the handle point 
$Q_1$ is moved, the opposite handle point $Q_2$ is recomputed by tracing a 
straightest geodesic from the anchor $P$ along the tangent direction defined 
from segment $Q_1 P$ to find the new position of handle $Q_2$. 

\paragraph{Rotation, Scaling and Translation.}
Our application also supports translation, rotation and scaling of a whole spline.
In the 2D settings these operations are obtained by just applying the same affine 
transform to all control points of a spline.
In the surface setting, we define the center of the transformation $C$ to be 
just the mesh point under the mouse pointer.

To apply the transformation, the normal coordinates of the control points are 
computed with respect to the center $C$, in a sort of discrete exponential map. 
Then, the linear transformation is applied on these 2D coordinates, 
which are finally converted back into mesh points by tracing straightest
geodesic paths outward from $C$. 

Translation needs special handling, as the center of the transformation $C$ 
is dragged to a new position $C'$. 
To compensate for the change of reference frame, the normal coordinates are rotated 
by the opposite angle of the parallel transport given by the tangent vector 
from $C$ to $C'$.

Note that, while the exponential map is not reliable 
to provide a dense map, we apply normal coordinates just to a relatively small
set of control points. In this case, we can tolerate the distortions caused by the 
curvature of the surface.

\subsection{Importing SVG drawings}
\label{sub:svg}

Sometimes, it may be convenient to map a whole 2D vector drawing, 
made of several primitives in the Euclidean plane, to the surface. 
Note that, unlike standard methods based on parametrization, 
we are not mapping the result of the drawing, but rather its control points: 
the final drawing is traced directly on the manifold, 
based on its vector specification, and can be further edited after mapping.  
See Figures \ref{fig:teaser} and \ref{fig:mapping} and the accompanying video for examples. 

\begin{figure}
  \centering
  \includegraphics[width=0.95\columnwidth]{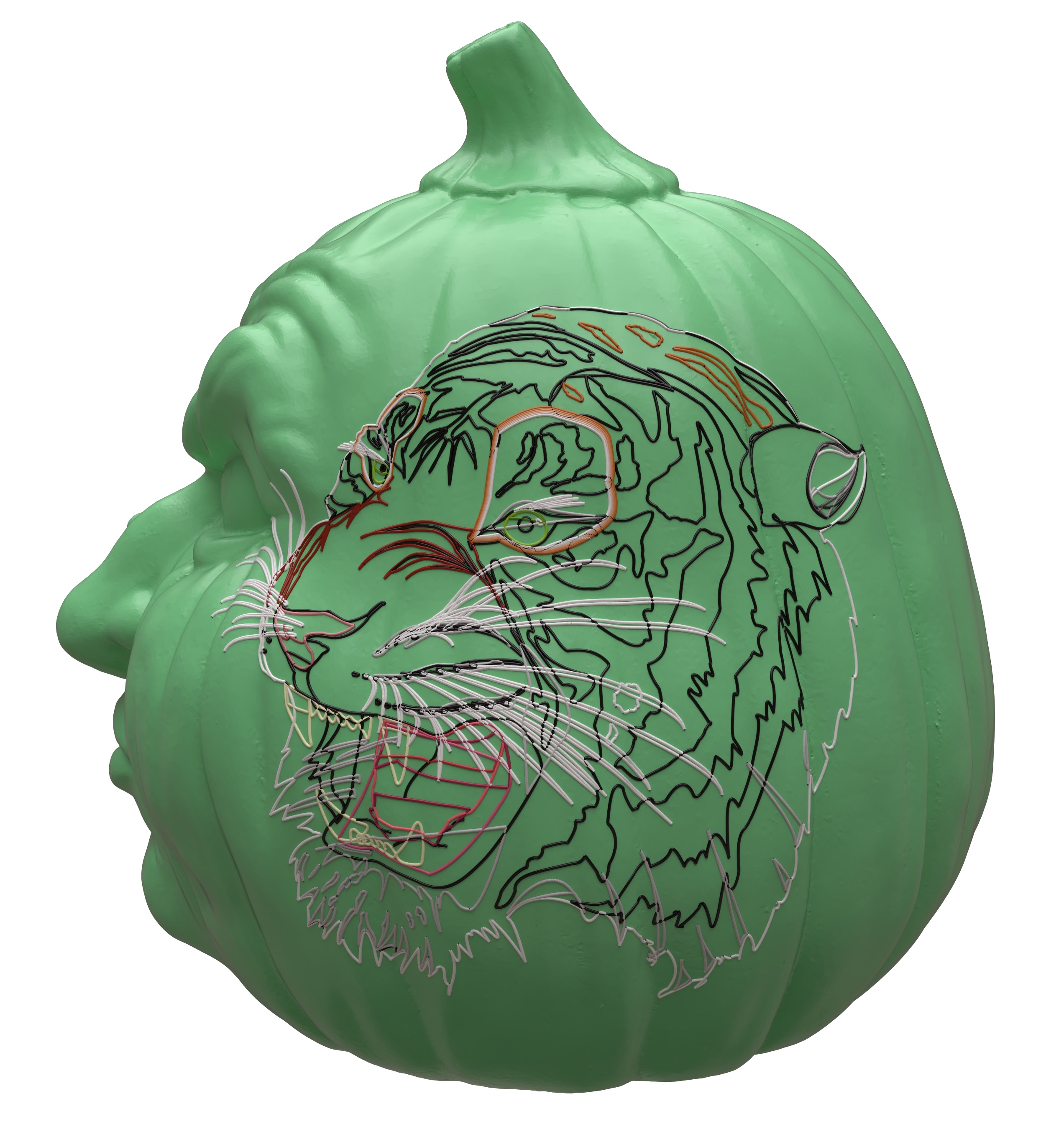}

 \caption{Example of importing a large SVG, made of 2056 curves, onto the 
 pumpkin model, consisting of 394k triangles. Our algorithm takes $289$ milliseconds to trace all curves.
 } 
 \label{fig:mapping}
  \Description{Image}
\end{figure}

This mapping is just meant to provide an initial placement of the 
control points on the target surface, allowing the user to adjust and 
fine tune the drawing afterwards. 
Therefore, we can allow for some distortion in the initial placement.

Our method is analogous to \cite{Biermann:2002bi}, and it is based on the 
conversion between polar coordinated in 2D and normal coordinates on the manifold. 
Each point of the SVG drawing is converted into a mesh point by taking its polar coordinates,
and tracing a geodesic from a center point in the given tangent direction,
for the given distance.

\section{Results and Validation}
\label{sec:results}

We validate our work by tracing curves over a large number of meshes, 
by comparing it with state-of-the-art solutions, and by 
performing interactive editing sessions, as shown in the accompanying video.
In summary, our algorithms produce a valid output in all trials, in a time compatible with 
interactive usage in over 99\% of the trials (Table \ref{tab:times}).
Overall, our method overcomes the limitations of state-of-the-art methods, 
producing valid results with any control polygon on any surface (Fig.~\ref{fig:WA});
and our running times are comparable (Table~\ref{tab:WA}) or faster (Fig.~\ref{fig:flipout}) 
than state-of-the-art methods. 

Concerning interactive usage, our system supports editing in all conditions 
for meshes of the order of one million triangles on a laptop computer. 
Interaction is still supported on meshes with several millions of triangles,
provided that single curves do not span too large a fraction of the model
(see, e.g., Figures~\ref{fig:teaser} and \ref{fig:mosaic}, Table~\ref{tab:mosaic-timings}, and the accompanying video).
Such cases are rare in actual editing sessions, as real designs are usually 
made of many splines, each consisting of several small segments.

\subsection{Robustness and performance}
\label{sub:thingi}

\begin{figure*}[tb]
  \centering
  \includegraphics[width=\textwidth]{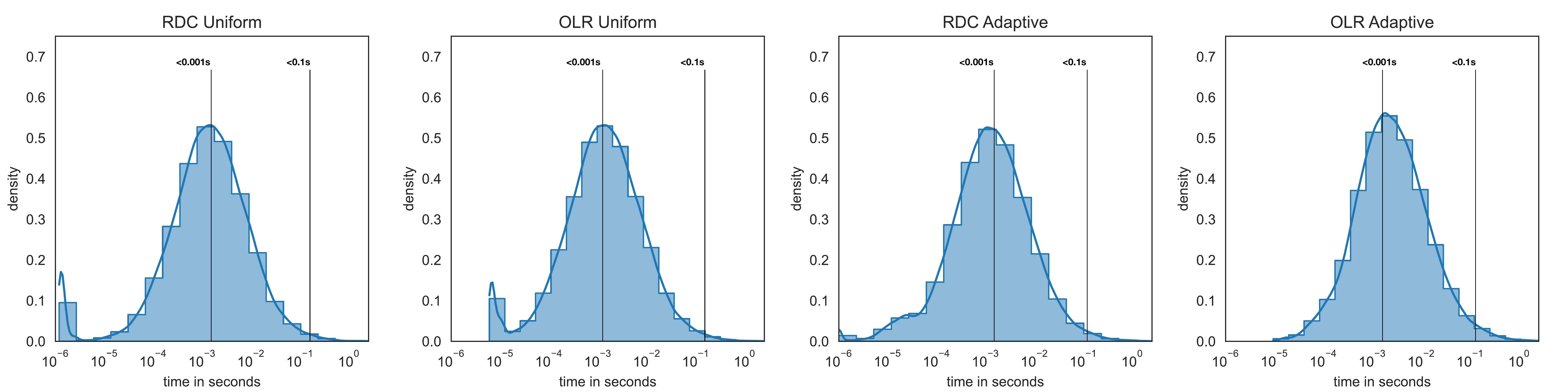}
  \caption{The distributions of running times of our four algorithms 
  for curve tracing in 556,700 trials on 5,567 models from the Thingi10k 
  repository, tracing 100 random curves on each model. 
  All algorithms provide a valid output in all trials. 
  The different algorithms have similar behavior and are compliant with 
  interaction (<0.1 seconds/curve) in more than 99\% of the trials. 
  For uniform subdivision, the OLR algorithm is slightly faster than the RDC algorithm; 
  while for adaptive subdivision, the RDC algorithm performs slightly better than the OLR algorithm. 
  Adaptive algorithms have slightly narrower distributions than uniform algorithms.
 } 
  \label{fig:bulk}
  \Description{Image}
\end{figure*}

\begin{table}
\begin{center}
\begin{tabular}{ccccc}
\toprule
 \multirow{2}{*}{algorithm} & \multicolumn{2}{c}{percent of trials} &  \multicolumn{2}{c}{times at percentile} \\
 &  < 0.001s &  < 0.1s & 90\% & 99\%\\
 \midrule
 RDC Uniform  & 43.1\% & 99.0\% & <0.0122 & <0.097\\
 OLR Uniform  & 44.7\% & 98.9\% & <0.0123 & <0.105\\
 RDC Adaptive & 43.9\% & 99.0\% & <0.0120 & <0.095\\
 OLR Adaptive & 30.0\% & 98.1\% & <0.0215 & <0.185\\  
 \bottomrule
 \end{tabular}
 \end{center}
  \caption{Time performances of our algorithms in 556,700 trials. 
  We report the percentage of trials in which tracing a curve takes less than 
  0.001 and 0.1 seconds, and the running times at the 90th and 99th percentiles, 
  respectively.  
 } 
  \label{tab:times}
\end{table}

We tested our algorithms for robustness by running a large experiment on 
the Thingi10k repository \cite{Thingi10K}. Our algorithm requires that the 
mesh is manifold and watertight, so we extracted the subset of meshes
that have those properties, for a total of 5567 models. The models are used 
as is, without any pre-processing.

For each model, we consider 100 random cubic curves.
For each curve, we take the model in its standard pose, and pick points on 
it by casting random rays orthogonal to the view plane, until we find 
four points that lie on the surface. These become the control points of the 
spline. We place no restriction in the arrangement of the control points. 
This gives us a total of more than half million control polygons.

For each test, we run both the RDC and the OLR tracing algorithms, in their 
uniform and adaptive configurations. The uniform RDC algorithm is expanded to
4 levels of recursion, which generates a geodesic polyline consisting of 48 
geodesic segments.
The uniform OLR algorithm is expanded to 6 levels of recursion, 
which generates a geodesic polyline consisting of 66 geodesic segments.
In fact, because of the different subdivision rules, we cannot generate the 
same number of segments for both schemes.  
For the adaptive variants, we set a threshold $\theta=5^{\circ}$ for the
maximum angle between consecutive geodesic segments along the polyline. 
In this case, the number of geodesic segments in output is variable, depending on the 
curve and on the method.
Since all algorithms generate very similar curves, 
the final tessellated paths that approximate the curve on the mesh,
consisting of one line segment per triangle crossed, have about the same 
number of segments in all four cases.

Trials were executed on a Linux PC with an AMD Ryzen 5 2600x and 32GB memory, 
running on a single core in all experiments. 
All our algorithms passed all the tests and generated curves that appear 
to be smooth in all cases. We quantitatively tested the smoothness of the generated
curves by numerically checking that the angle formed by consecutive geodesic segments
of the polyline was always below the $\theta=5^{\circ}$ threshold and tested the continuity
by checking that the distance between consecutive points was always smaller than the length
of the longer edge in the mesh.

In  Table \ref{tab:times} and Fig.~\ref{fig:bulk}, we compare the timing 
performance of the four algorithms. All algorithms perform quite similarly,
and remain interactive in all cases, with roughly 40\% of trials running at 
less than 1 millisecond per curve, and 99\% of the trials running faster 
than 0.1 second/curve. 
The few trials in which they take more time are concerned, with very few exceptions, 
either with very long curves on large meshes (>1M triangles), or with meshes containing 
many topological holes, in which finding shortest paths between points is more expensive.

There are small differences in the performances of the different algorithms.
For uniform subdivision, the OLR algorithm results as fast as the 
RDC algorithm, beside generating a more refined geodesic polyline.
For adaptive subdivision, the RDC algorithm runs slightly 
faster than the OLR algorithm. 
These differences are probably due to the simpler structure of the OLR uniform 
algorithm in one case, and to the more involved structure of the OLR adaptive 
algorithm in the other.
In fact, both variants of the RDC algorithm follow the same recursive pattern. 
On the contrary, the uniform OLR algorithm expands the curve level by level, 
following a simpler pattern; while the OLR adaptive algorithm requires a recursive 
pattern, with a slightly more involuted structure than the RDC algorithms. 
 
For the sake of brevity, we do not present here results on the algorithms 
for curve tracing and point insertion, which run much faster than the tracing 
algorithms. 

In the previous experiments, the cost of computing a curve depends on both the 
length of the curve and the size of the mesh, with trends that are not linear. 
Roughly speaking, the cost of finding the initial path of geodesics depends on 
both the length of the curve and the size of the mesh, while the subsequent 
cost of finding the shortest path depends just on the length of the curve. 
As the relative length of the curve grows, the cost of finding the initial 
path prevails, since it may requires exploring most of the mesh. 
Statistics on the relative costs of the two phases are shown
in Fig.~\ref{fig:path-computation}.

\subsection{Sensitivity to the input mesh}
\label{sub:sensitivity}

All the algorithms presented in Sec.~\ref{sub:geodesicprimitives} are driven 
by the connectivity of the underlying mesh.
In particular, all intersections between the traced lines and the mesh are 
computed locally to each triangle and forced to lie on its edges, so that each 
traced line consistently crosses a strip of triangles.
With this approach, we could process even meshes containing nearly 
degenerate triangles, with angles near to zero and edge lengths near to the 
machine precision, by relying just on floating point operations, 
without incurring in numerical issues. 
While this is usually not the case with models used in a production environment,
such kind of meshes is are common in the Thingi10k repository and 
provides a stress test for the robustness of our algorithms.

On the other hand, our algorithm for point-to-point shortest path assumes the
initial guess obtained during Phase (i) to be homotopic to the result.
This assumption is common to all algorithms for computing locally shortest paths
\cite{sharp2020flipout}, and it is reasonable as
long as the mesh is sufficiently dense and uniform with respect to the underlying
surface. If, conversely, the mesh is too coarse and anisotropic, then Phase (i) may provide an
initial guess, which cannot be homotopically shortened to the correct solution.
In this case, a naive application of the algorithm may get stuck in local 
minima of the space of shortest paths, leading to a wrong curve.

This limitation is quite rare in practice for meshes used in design applications,
which is our target, but did happen for some meshes in the Thingi10k dataset.
We overcome this limitation without changing the algorithm itself, but simply by
creating a more accurate graph for computing the initial guess when dealing with
meshes with long edges.

When we build the dual graph to be used in Phase (i), we split mesh edges that are too long at their midpoint, 
until all edges are shorter than a given threshold, 
and we symbolically subdivide their incident triangles accordingly. 
Note that this subdivision is done just for the purpose of building the graph, 
without changing the underlying mesh.
In this augmented graph, a single triangle
may be represented by multiples nodes, giving us a more accurate approximation of paths. 
This approach has the effect of densifying the graph
without changing the mesh upon which we run Phase (ii).
We chose the 5\% of the diagonal of the bounding box of the model as threshold.
Once the strip is computed on the augmented graph, we reconstruct the strip
on the mesh using the graph's node provenance, i.e. the mesh triangle 
corresponding to each node, which we store during initialization.

An alternative approach to cope with the same problem would be to pre-compute
an intrinsic Delaunay triangulation in the sense of \cite{Sharp:2019ea} and do
all computations by using intrinsic triangulations.
We did not adopt this latter solution since the problem occurs quite seldom,
while using intrinsic triangulations would require more complex data structures.


\begin{table}[tb]
\begin{center}
\begin{tabular}{l@{}rrrr}
\toprule
 \multicolumn{2}{c}{model} &  \multicolumn{2}{c}{WA} & ours (OLR)\\
name &  triangles &  pre-proc. (s) & tracing (ms) & tracing (ms) \\
 \midrule
 cylinder &10k & 54 & 2--2 & 1--1 \\
 kitten & 37k & 234 & 3--3 & 3--3 \\
 bunny & 140k & 665 & 2--2 & 10--12 \\
 lion & 400k & 2316 & 3--3 & 4--24 \\
 nefertiti & 496k & 2571 & 6--64 & 25--67 \\  
 \bottomrule
 \end{tabular}
 \end{center}
  \caption{Compared time performances of curve tracing with the WA method and
  our OLR, on the curves shown in Fig.~\ref{fig:WA} and Fig.~\ref{fig:WA-cylinder}.
  Each curve is sampled at 67 points, including endpoints; curve tracing times are
  averaged on each curve repeating tracing 1000 times per curve, and we report
  minimum and maximum times over the different curves shown in the images.    
 } 
  \label{tab:WA}
\end{table}

\begin{table}[tb]
\begin{center}
\begin{tabular}{l@{}rrrrr}
\toprule
\multicolumn{2}{c}{model} & \multicolumn{1}{c}{control} & \multicolumn{1}{c}{subdivided} & \multicolumn{2}{c}{time (ms)} \\
name & triangles & \multicolumn{1}{c}{polygons} & \multicolumn{1}{c}{segments} & \multicolumn{1}{c}{total} & \multicolumn{1}{c}{per curve} \\
\midrule
veil & 132k & 2 & 402 & 2.3 & 1.1 \\ 
arm & 145k & 2 & 856 & 35.6 & 17.8 \\ 
boot & 175k & 2 & 755 & 21.1 & 10.5 \\ 
deer & 227k & 4 & 1511 & 21.8 & 5.4 \\ 
lady & 281k & 9 & 1917 & 11.4 & 1.2 \\ 
car & 282k & 2 & 670 & 28.0 & 14.0 \\ 
pumpkin & 394k & 5 & 1750 & 30.0 & 6.0 \\ 
rhino & 502k & 7 & 2395 & 39.8 & 5.6 \\ 
owls & 641k & 14 & 3224 & 20.8 & 1.4 \\ 
alexander & 699k & 5 & 1560 & 20.5 & 4.1 \\ 
vase & 754k & 8 & 1677 & 9.0 & 1.1 \\ 
nike & 5672k & 7 & 4147 & 253.8 & 36.2 \\ 
\midrule
nefertiti & 496k & 463 & 64110 & 73.4 & 0.2 \\ 
dragon & 7218k & 221 & 60656 & 761.7 & 3.4 \\ 
\bottomrule
\end{tabular}
\end{center}
\caption{Time performances for curve tracing on the models in Fig.~\ref{fig:mosaic}
and in the teaser, using the uniform OLR algorithm with 5 levels of subdivisions.
We report the total time of computing all the curves and the average time of
computing a single curve.
For all the reported models, our algorithm achieves performance compatible with
real-time editing, since the time per curve is at most in the order of tens of 
milliseconds.
}
\label{tab:mosaic-timings}
\end{table} 

\begin{figure}[tb]
  \centering
  \includegraphics[width=\linewidth]{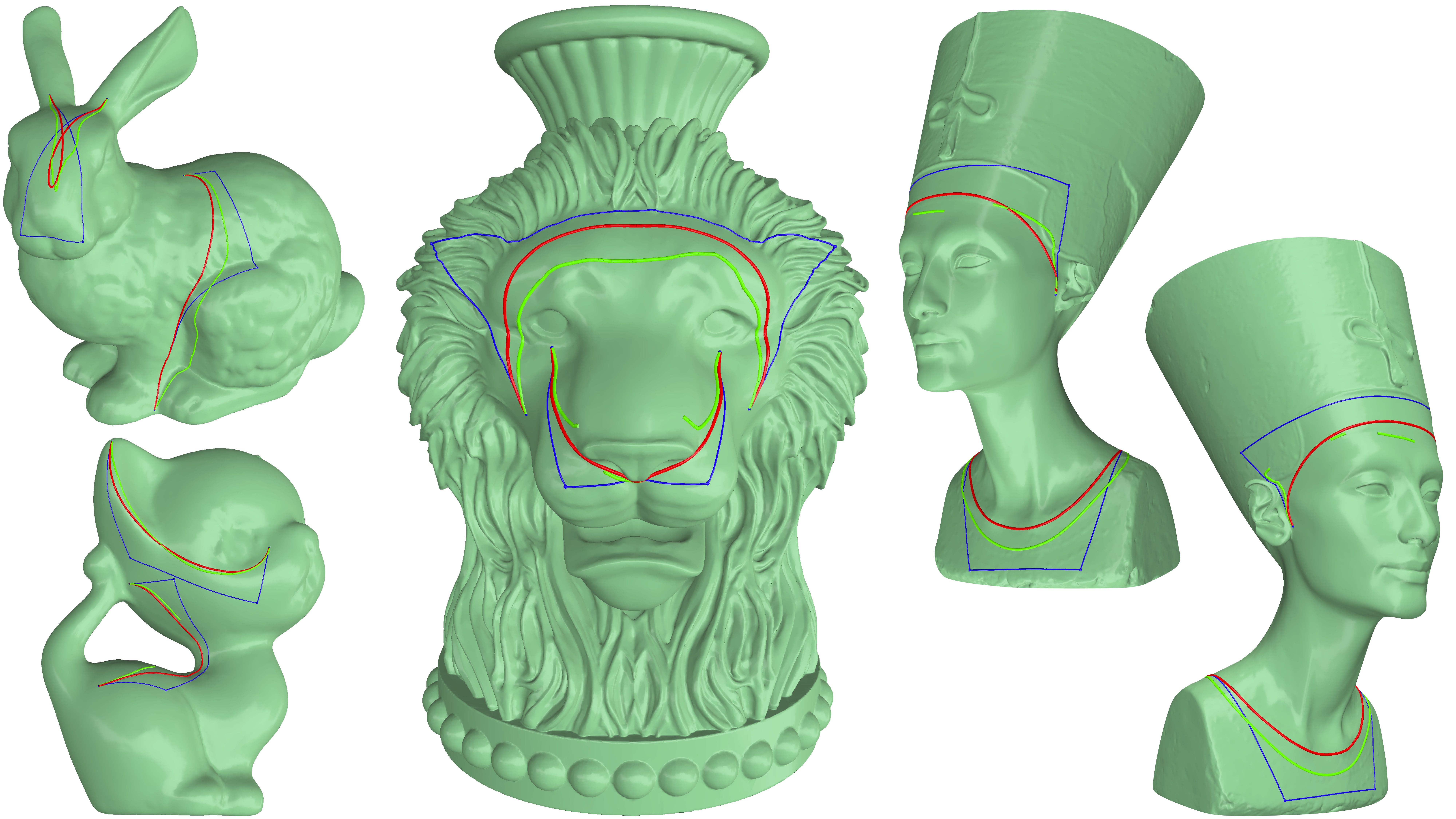}
  \caption{Comparisons between the WA method (green curves) on ours (red curves);
  control polygons in blue. 
  The WA curves may contain heavy artifacts (bunny), lose tangency at the
  endpoints (bunny, nefertiti), or be  broken (kitten, lion, nefertiti).}
  \label{fig:WA}
  \Description{Image}
\end{figure}

\begin{figure}[tb]
  \centering
  \includegraphics[width=\linewidth]{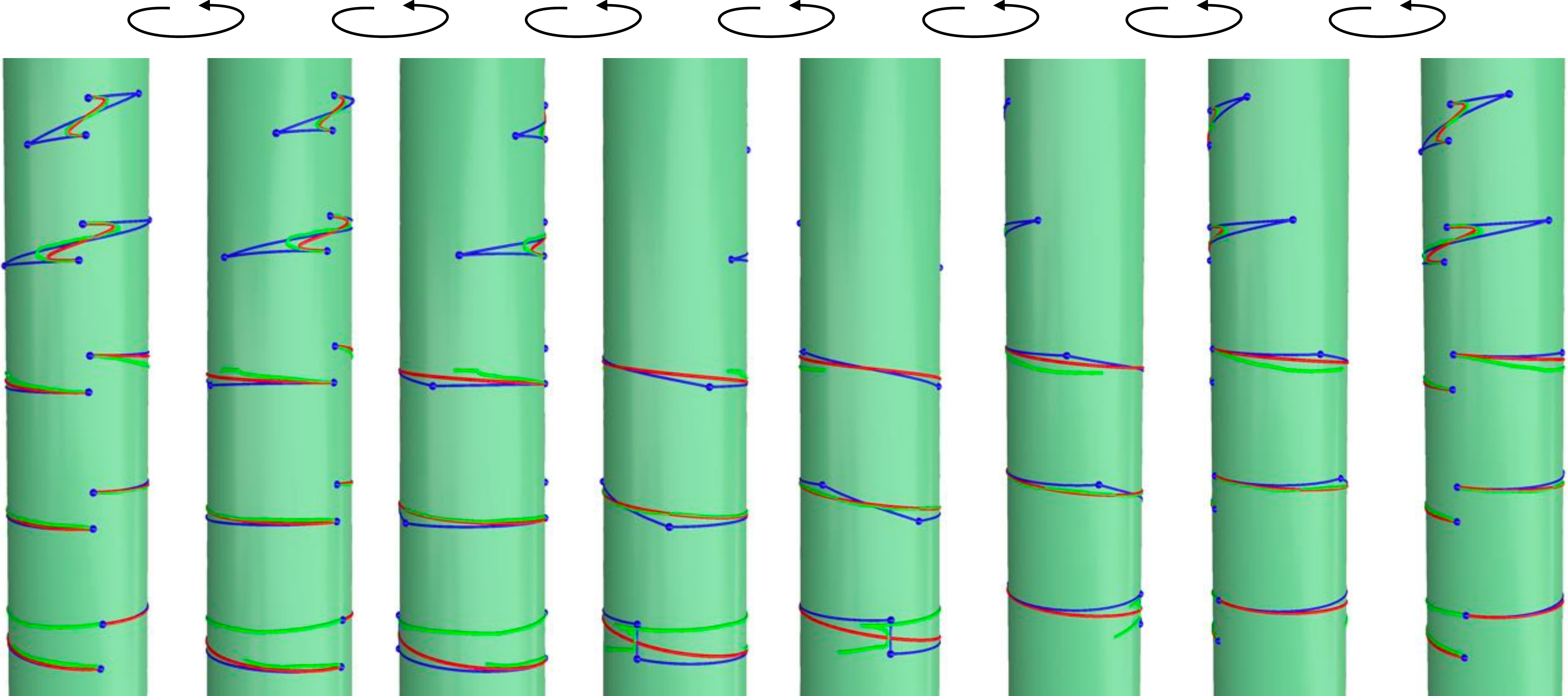}
  \caption{Evolution of a curve while dragging handle points about a cylinder
  (top to bottom, rotated views left to right) with the WA method (green curves)
  and ours (red curves). Our curve jumps from the ``reversed S'' configuration
  to the spire and remains stable throughout. The WA curve is stable only in 
  the ``reversed S'' configuration, next it breaks, then it forms a spire, 
  and eventually it breaks again.}
  \label{fig:WA-cylinder}
  \Description{Image}
\end{figure}

\subsection{Comparison with the state-of-the-art}
\label{sub:flipOut}

\paragraph{Weighted Averages (WA) \cite{Panozzo:2013eh}.}
The method presented in \cite{Panozzo:2013eh} tries to estimate the RCM on a
surface, by approximating the geodesic distance on the input mesh $M$ with the
Euclidean distance on a higher-dimensional embedding of $M$. 
Given a set of control points and weights, instead of resolving 
Eq.~\ref{eq:RCM} (of this paper) on $M$, they compute the standard affine
average of Eq.~\ref{eq:affine} in the embedding space. Then they use a special
technique, called \emph{Phong projection}, to bring the resulting space curve to the embedded
mesh. Finally they recover the corresponding points on $M$.
We compare with this technique by using the implementation provided by the
authors, with the same sampling used in our experiments.

The embedding and the data structures to support Phong projection are computed
in a pre-processing step, which is quite heavy in terms of both time and space,
and can hardly scale to large datasets (see Table \ref{tab:WA}). 
We managed to pre-process datasets up to about 500K triangles, but we could 
not process some of the larger datasets we use in our work, because memory 
limits were exceeded.   
The embedding is built by sampling a small subset of the vertices first 
(fixed to 1000 by the authors), computing all-vs-all geodesic distances on $M$
for such subset, and embedding such vertices in a 8D Euclidean space by keeping
their mutual Euclidean distances as close as possible to their geodesic 
distances on $M$. 
The remaining vertices are embedded next, by using the positions of the first
embedded vertices as constraints. The connectivity of $M$ is preserved, and the
positions of vertices are optimized, so that the distances between adjacent
vertices remain as close as possible to their distances on $M$.    

The online phase of WA is very fast, and it is insensitive to the size of the
input and the length of the curve (see Table \ref{tab:WA}). 
However, we experienced a case that took one order of magnitude more time than the others. 
We conjecture this is due to some unlucky configuration for the Phong projection, slowing its convergence.
On the contrary, the performance of our method is dependent on both the size of the dataset and
the length of the curve, being faster than WA on small datasets and shorter curves,
and slower on large datasets and long curves. 
In terms of speed, both methods are equally compatible with interaction on the tested models. 

Concerning the quality of the result, the smoothness of the WA embedding, which is necessary to guarantee the
smoothness of the Phong projection, cannot be guaranteed, hence the WA method
suffers of limitations similar to the RCM method analyzed in Sec.~\ref{sub:karcher}.
As soon as the segments of the control polygon become long, relevant artifacts
arise, and the curve may even break into several disconnected segments. 
Some results obtained with the WA method, compared with our results, are shown
in Figures \ref{fig:WA} and \ref{fig:WA-cylinder}.
In particular, Fig.~\ref{fig:WA-cylinder} exemplifies the behaviors of the
two methods as a control polygon becomes larger.
While our curve remains smooth and stable throughout, except for the necessary
jump between the ``reversed S'' and the spire, the WA curve becomes unstable 
and breaks in most configurations where the control points are far apart. 

\begin{figure}[tb]
  \centering
  \includegraphics[width=\columnwidth]{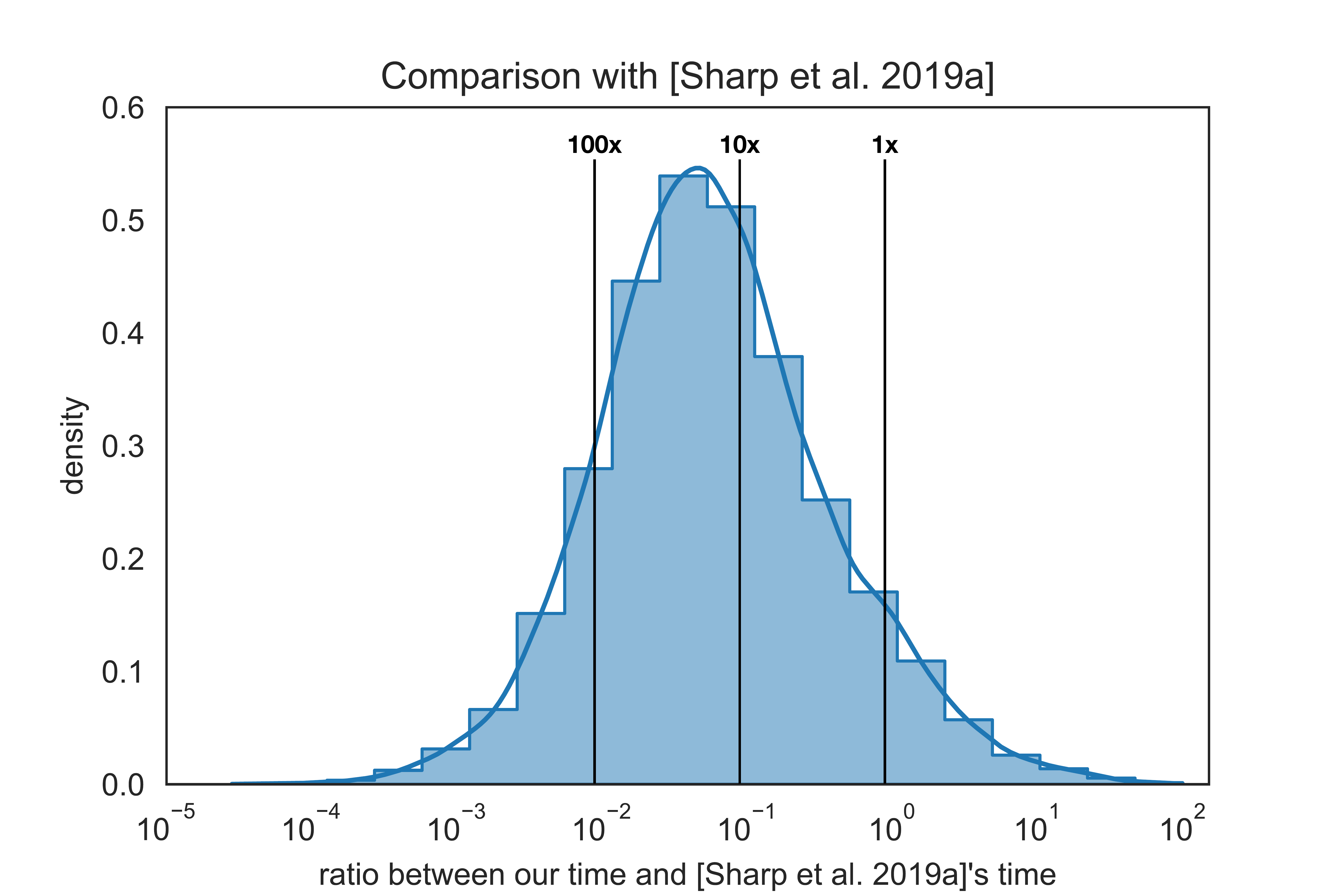}
  \caption{The graph shows the distribution of the ratio of the running times 
  between our RDC uniform algorithm and the implementation from \cite{geometrycentral},
  which is based on  the \emph{flipOut} method for computing geodesics \cite{sharp2020flipout}.
  Here we report only the 78,854 trials, out of 556,700, for which 
  \cite{geometrycentral} could provide an output. 
  On average, our RDC algorithm implementation provides more than a 10x speedup over \cite{geometrycentral}.
} 
  \label{fig:flipout}
  \Description{Image}
\end{figure}

\paragraph{RDC based on flipOut \cite{geometrycentral, sharp2020flipout}.}
The \emph{flipOut} algorithm was proposed recently \cite{sharp2020flipout} 
as a fast solution to the computation of locally shortest geodesic paths. 
On the basis of the \emph{flipOut} algorithm, the same authors have implemented 
the algorithm of \cite{Morera:2008jk}, which uses the same recursive scheme of our 
RDC algorithm for curve tracing. 

While our algorithms have no limitations, and could provide 
a valid output in all 556,700 trials, the algorithm in \cite{geometrycentral} 
requires that the control polygon does not contain
self-intersections, a case which is pretty common with cubic curves,
and happens in 33\% of the randomly generated polygons.
This is due to an intrinsic limitation of the \emph{flipOut} algorithm, which was discussed in \cite{sharp2020flipout}.

We have used the implementation provided by the authors \cite{geometrycentral} 
to run the same experiments of Sec.~\ref{sub:thingi}, 
with the same parameter used for our RDC algorithm with uniform expansion.
Because of the above limitation, we excluded from the comparison all the trials for which
the algorithm of \cite{geometrycentral} could not provide an output, keeping a total of 78,854 out of 
556,700 trials.

From a visual inspection of random samples of the results, it seems that both
algorithms generate the same curves. 
In Fig.~\ref{fig:flipout}, we present a comparison between the performances
of the two algorithms. 
Our RDC uniform algorithm exhibits a speedup of more than 10x on average.
This speedup seems to be due to a faster estimate of the initial guess to compute the point-to-point
shortest paths.
In fact, both the algorithm we presented in Sec.~\ref{sub:geodesicprimitives} and the \emph{flipOut} algorithm 
start from an initial path to iteratively shorten it and converge to the shortest path. 
The algorithm adopted in \cite{geometrycentral} to compute the initial path is much slower 
in finding this initial guess, while \emph{flipOut} is comparable to ours in the shortening step. 
This is shown in Fig.~\ref{fig:path-computation}.

\begin{figure*}[tb]
  \centering
  \includegraphics[width=0.24\textwidth]{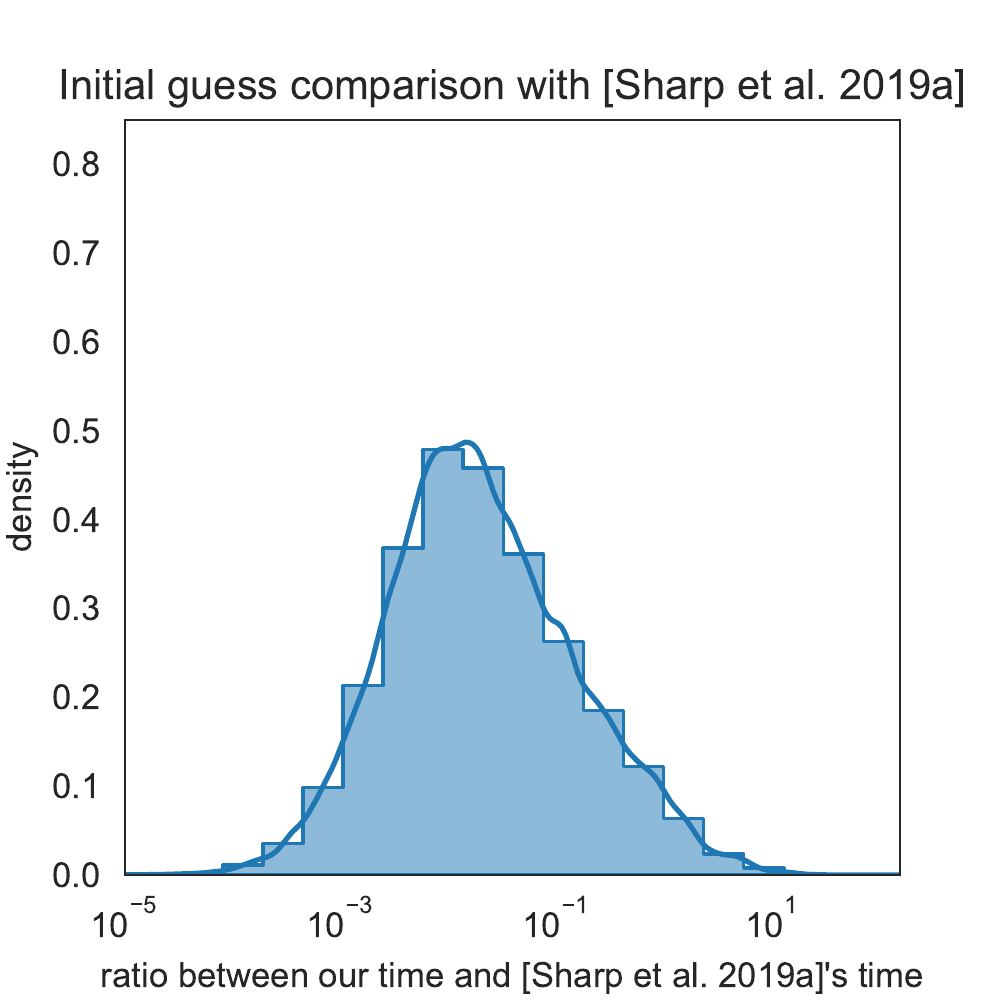}
  \includegraphics[width=0.24\textwidth]{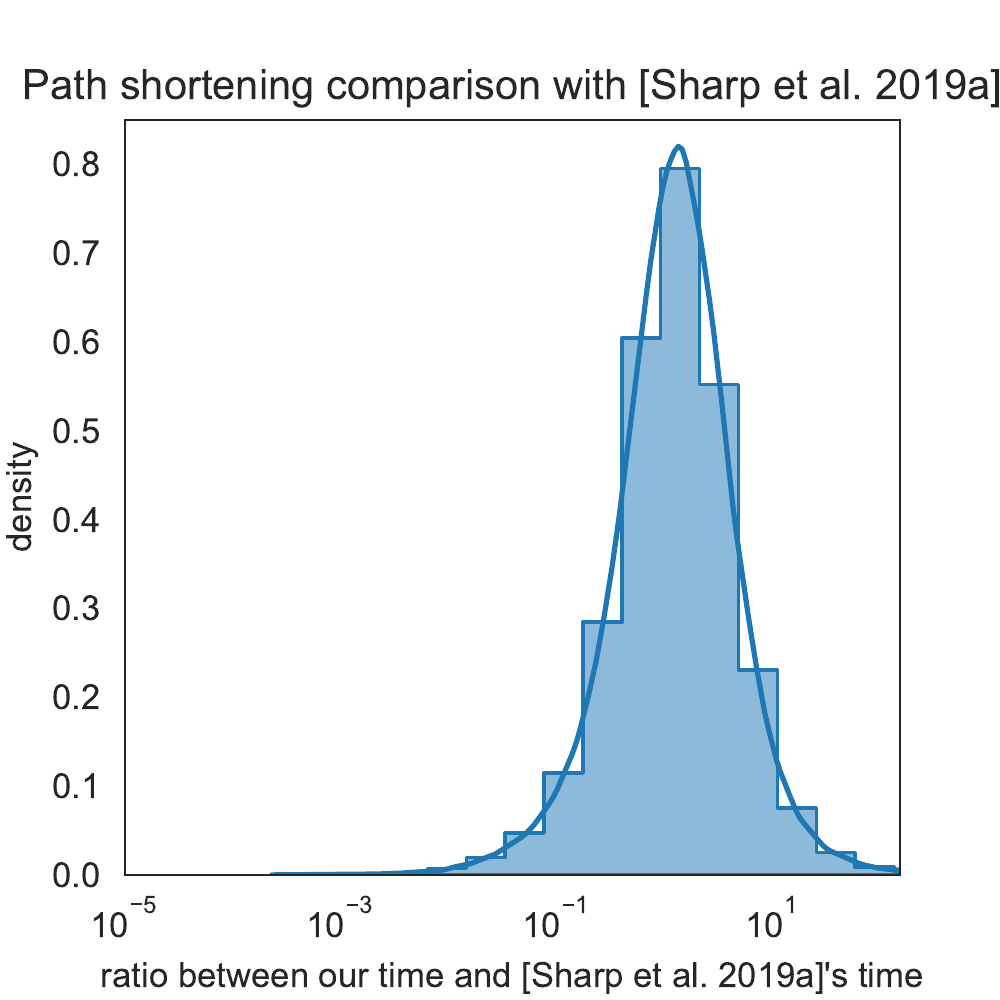}
  \includegraphics[width=0.24\textwidth]{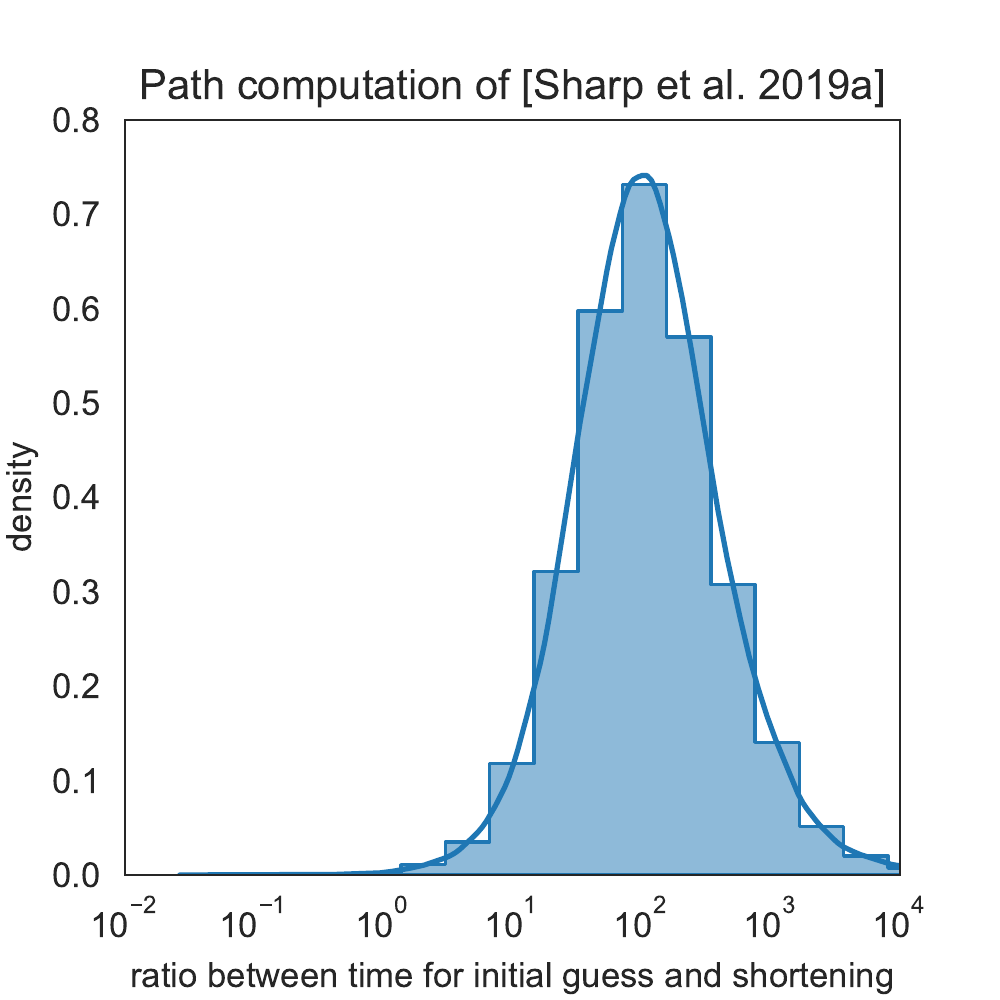}
  \includegraphics[width=0.24\textwidth]{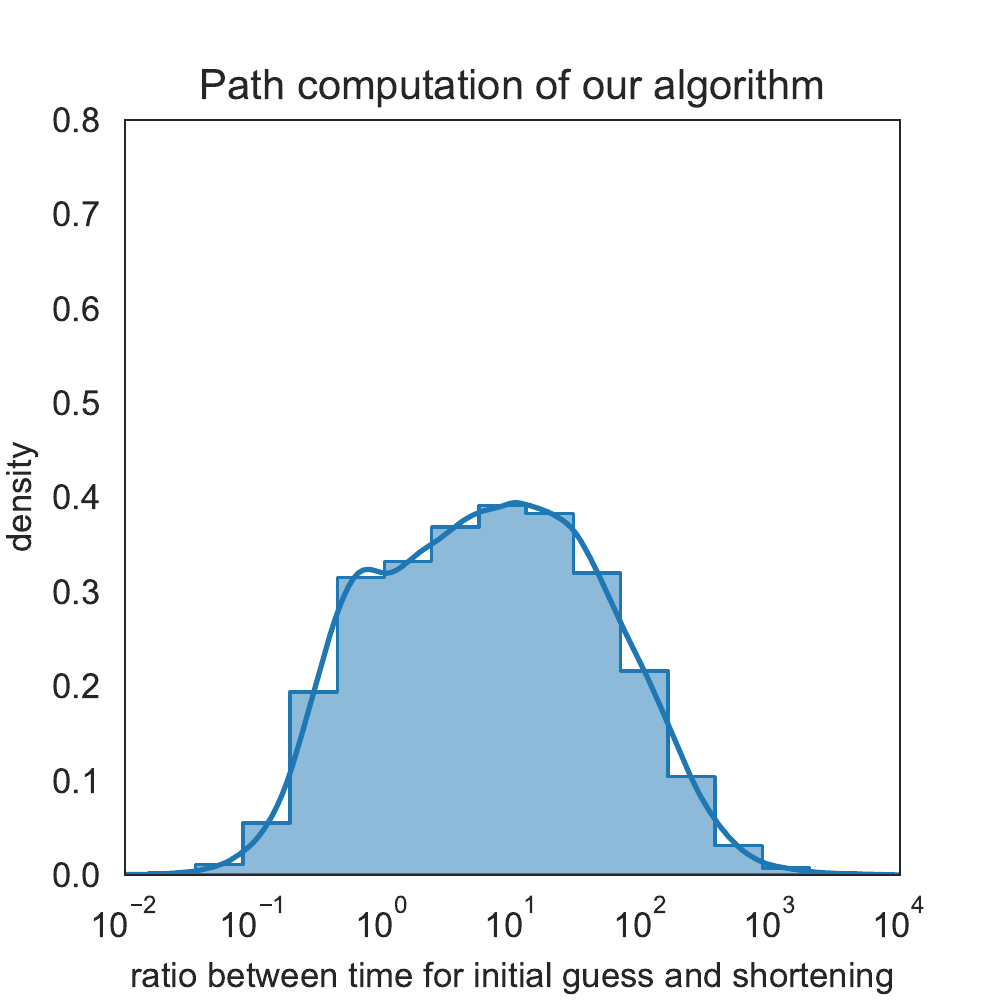}
  \caption{Left (2 charts): When computing the initial guess, our implementation is almost 
  always one order of magnitude faster than \cite{geometrycentral}, while the
  second shortening phase are comparable in performance. Right (2 charts): On average, our algorithm
  spends 10 times more time in computing the initial strip then in refining
  the final result; while \cite{geometrycentral} is two orders of magnitude slower in
  completing the first phase, compared to the second one.}
  \label{fig:path-computation}
  \Description{Image}
\end{figure*}

\begin{figure*}[tb]
  \centering
  \includegraphics[width=\textwidth]{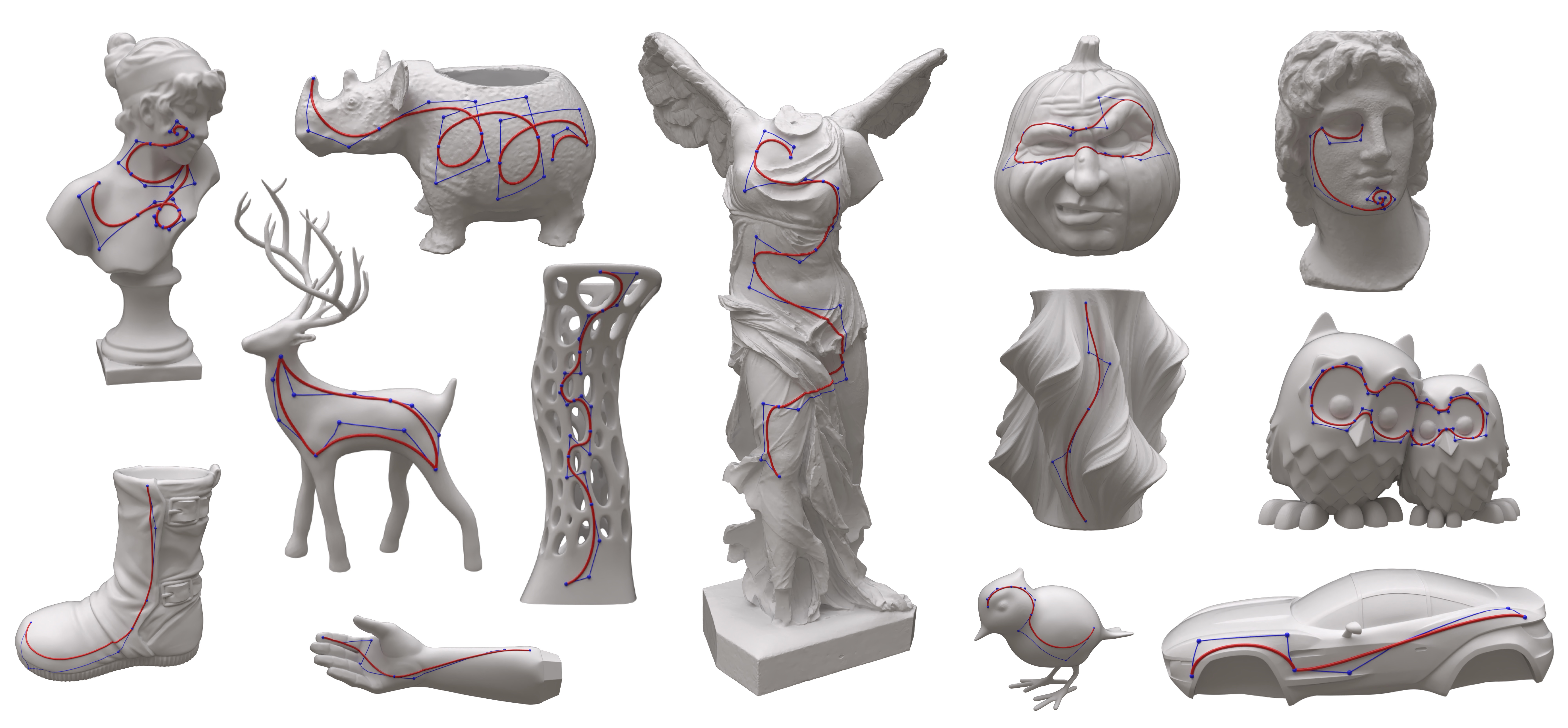}
  \caption{A gallery of models and splines drawn with our method. 
  Both  smooth ($C^1$) and corner ($C^0$) continuity at junction points are exemplified.
  The selected models span a wide range of shapes and the sizes of meshes 
  vary between about 130k and 5.7M triangles.} 
  \label{fig:mosaic}
  \Description{Image}
\end{figure*}

\subsection{Interactive use}
\label{sub:stress}
We have used extensively our system on a variety of models. All editing sessions
where performed on a MacBook laptop with a 2.9GHz Quad-Core Intel Core i7 with 
16GB memory, running on a single core.

Fig.~\ref{fig:mosaic} presents a gallery of curves drawn interactively on 
objects picked from the Thingi10k collection. Statistics for each example
are summarized in Table~\ref{tab:mosaic-timings}.
Interaction is quite intuitive, being supported with a GUI that mimics the
drawing of spline curves in standard 2D systems, as described in
Sec.~\ref{sub:user}. The most tricky aspects, with respect to the standard 2D
case, are concerned with using tangents that consist of geodesic lines instead
of straight lines. In our experience, the use of geodesic tangents, which is
intrinsic to the manifold metric, becomes intuitive quickly.

We have stressed our system by working on very large meshes as shown in 
Fig.~\ref{fig:teaser}. Even on meshes of a few million triangles, 
our implementation remains interactive, as shown in Table~\ref{tab:mosaic-timings}.




\section{Concluding remarks}
\label{sec:conclusions}

We propose methods for interactively drawing and editing of Bézier curves 
on meshes of millions of triangles, without any limitation on the curve shape 
and extension of control polygons.
Our algorithms are robust, having been tested on over five thousands 
shapes with over half a million randomly generated control polygons, and they 
are compatible with interactive usage even on large meshes. 
Our new Open-uniform Lane-Riesenfeld scheme provides the smoothest practical 
solution so far for Bézier curves in the manifold setting; while our 
De Casteljau variation is simple to implement, at the price of less smoothness.

The main limitation of these methods lie in the discontinuities of the space of 
curves with respect to their control points: curves are always smooth, but they may make jumps during editing. 
Such a discontinuity is inherent of the geodesic metric, and it can be overcome by using a spline with shorter control polygons, instead of a single large polygon, to define the curve. 
Our algorithms for point insertion greatly help in this task.  

In the future, we want to consider other types of splines.
An extension of our approach to B-splines seems straightforward.
An extension to interpolating splines seems easy, but it requires manifold 
extrapolation, which may become unstable. The most complex extension would be
to handle NURBS, which at this point remains unclear how to do. 
More generally, the smoothness analysis in the non-uniform case needs a 
thorough investigation.  


\section*{Acknowledgments}
We wish to thank Chiara Eva Catalano, Tom Duchamp, Kai Hormann, Daniele Panozzo, and Giuseppe Patan\'e for helpful discussions; Marzia Riso for her help with experiments and figures; and Michele Serpe for modeling a mesh for us. Models other than Thingi10k: Nefertiti is courtesy of Scan the World; Armadillo and Bunny are courtesy of the Stanford 3D Scanning Repository.

\appendix
\section{Geodesics on Surfaces and Meshes}
\label{app:geodesic}


In the following, we provide just a summary of basic concepts. 
For a complete account on this subject, we refer to \cite{do1992riemannian}(smooth setting) 
and \cite{Crane:2013:DGP} (discrete setting).

\paragraph{Smooth setting.}
Let $M$ be a smooth surface embedded in $\mathbb{R}^3$.
The embedding induces a Riemannian metric on $M$, 
defining the length $L(\gamma)$ of any parametric curve $\gamma$ on $M$. 
Without lack of generality, we consider $\gamma$ to be parametrized in $[0,1]$ with constant speed. 

The \emph{geodesic distance} $d(p,q)$ between any two points $p, q \in M$ is the infimum of length $L$ over all curves $\gamma$ such that $\gamma(0) = p$ and $\gamma(1) = q$;
and one such curve $\bar{\gamma}$ satisfying $L(\bar{\gamma})=d(p,q)$ is called a \emph{shortest geodesic path} between $p$ and $q$. 

A shortest geodesic path may not be unique. 
For any $p\in M$, the set of points $q$ such that there exist more than one geodesic shortest path connecting $p$ to $q$ belong to the \emph{cut locus} of $p$. 
If $q$ stays away from the cut locus of $p$, then both the geodesic distance $d(p,q)$ and the (unique) shortest path $\gamma_{p,q}$ joining them vary smoothly with $q$.
On the contrary, if $q$ moves along a trajectory that crosses the cut locus of $p$, then the distance $d(p,q)$ changes with continuity, but not smoothly, at the cut locus; while $\gamma_{p,q}$ may jump to a totally different curve. 
These facts are relevant to assess the stability of the methods discussed in Section \ref{sec:method}. 

A path $\gamma_{p,q}$ is said to be \emph{locally shortest} if there exist $\delta>0$ such that for any interval $[t-\delta,t+\delta]\subseteq[0,1]$ the restriction of $\gamma_{p,q}$ to  $[t-\delta,t+\delta]$ is a shortest path. 
A curve connecting $p$ to $q$ can be homotopically deformed to a locally shortest path; this is at the basis of all algorithms for computing locally shortest paths, including the algorithm we present in Sec.~\ref{sub:geodesicprimitives}.    

Geodesic curves can be also characterized by their \emph{straightness}. 
In order to assess the curvature of lines in the intrinsic geometry of $M$, one needs to introduce the covariant derivative, which we omit here for brevity. 
Intuitively, from an extrinsic point of view, a geodesic curve $\gamma$ does not make any further turn except the strictly necessary to follow the curvature of $M$: it turns \emph{with} $M$, but it does not turn \emph{on} $M$. 
Thus, geodesics play the role of straight lines on $M$. 

A geodesic curve is completely defined by its starting point and tangent vector.  
Following this observation, the \emph{exponential map} $\exp_p : T_p M \longrightarrow M$ maps vectors of the tangent plane $T_p M$ at $p$ to points on the surface. In general, the exponential map is not injective. 
A neighborhood $U$ of $p$ over which $\exp_p$ is invertible is called a \emph{normal neighborhood} of $p$; 
the inverse of the exponential map on $U$ defines the \emph{logarithmic map} $\log_p: U \to T_p M$, which provides local coordinates around $p$ called \emph{normal coordinates}. 

A set $U\subset M$ is said to be a \emph{totally normal neighborhood} if it is a normal neighborhood for all its points; and it is said to be \emph{strongly convex} if it contains all shortest paths between pairs of its points.
The maximum radius that a region $U$ can be extended to, while remaining a [totally] normal neighborhood, or a convex set, depends on the Gaussian curvature of $M$ and is not easy to assess. 
For this reason, methods based on exp and log maps cannot guarantee robustness in the general case. 

\paragraph{Discrete setting.}
If $M$ is a polyhedral manifold -- which, w.l.o.g., we can consider to be a triangle mesh -- then it is no longer smooth. 
The concepts of shortest geodesic path, geodesic distance and locally shortest path are nonetheless well defined.
A geodesic path $\gamma$ between two points $p, q$ on $M$ is a polyline, intersecting a strip of triangles of $M$.  
As long as $\gamma$ does not cross any vertex of $M$, the segments of this polyline can be obtained by flattening the strip of triangles to the Euclidean plane.  

Geodesics that cross vertices are more complex to handle.
A vertex $v_i\in M$ can be classified depending on the sign of its \emph{angle defect}, defined as the difference between $2\pi$ and the total angle about vertex $v_i$. 
It turns out that no (locally) shortest path can pass through a vertex with positive angle defect;  
while infinitely many shortest paths can reach a vertex with negative angle defect 
from a given direction and take divergent directions to the other side of it. 
This fact makes the definition of a \emph{straightest geodesic} more complicated than in the smooth setting. 
Following the definition of \cite{polthier1998straightest},  a straightest geodesic intersecting a vertex $v$ is required to bisect the total angle about $v$. 
With this definition, a straightest geodesic $\gamma$ is also locally shortest if and only if it does not cross any vertex with positive angle defect.

\section{Open-uniform Lane-Riesenfeld subdivision}
\label{app:subLR}
In the  Euclidean setting, a B\'ezier curve with control polygon $\Pi_k$ can be rewritten as an open-uniform B-spline of degree $k$ (order $k+1$) with the same control polygon $\Pi_k$ and knot vector $(00\ldots011\ldots1)$, where the $0$ and $1$ are repeated $k+1$ times \cite{Salomon:2006wp}.
%
The generalized Oslo algorithm \cite{Lyche:93} 
is applied to repeatedly insert knots at the midpoints of all non-zero intervals in the knot vector, producing a sequence of open uniform B-splines, all describing the same curve, whose control polygons $\Pi_{LR}^n$ converge to curve itself. 
In the following, we sketch the construction for a generic value of $k$, and we provide the complete solution for values $k=2, 3$, which are most relevant in the applications.  
See \cite{Cashman:2007jh} for further details in the Euclidean setting.

Let $\Pi^0_{LR} = \Pi_k$ be our initial control polygon, and let $\Pi^{n+1}_{LR}$ be the polygon obtained from $\Pi^n_{LR}$ with one round of knot insertions.
The knot vectors associated to the first levels of the subdivision are:
\[
\begin{array}{c}
(0 0 \ldots 0 1 1 \ldots 1)\\
(0 0 \ldots 0 1 2 2 \ldots 2)\\
(0 0 \ldots 0 1 2 3 4 \ldots 4)\\
(0 \ldots 0 1 2 3 4 5 6 7 8 \ldots 8)\\
\cdots\cdots
\end{array}
\]
where the first and the last node are always repeated $k+1$ times, and nodes are renumbered at each level. 
Note that, at level $n$, polygon $\Pi^n_{LR}$ contains $2^n+k$ control points, and its associated knot vector contains $2^n+1+2k$ nodes, with $2^n$ non-null intervals.
Each point of $\Pi^{n+1}_{LR}$ is computed by applying the Oslo algorithm, as an affine average of at most $k+1$ consecutive points of $\Pi^n_{LR}$, with weights depending on $k+2$ consecutive knots of the knot vector of level $n$. 
The first few levels of the subdivision need special stencils, obtained directly from knot insertion; as soon as $n\geq\lceil\log_2 (k+1)\rceil$, the knot vector contains at least $k+1$ non-null intervals, and the subdivision rules stabilize: we have the standard stencils of the uniform Lane-Riesenfeld subdivision \cite{Lane:dq} for control points that depend just on uniform nodes, plus a set of $2k-2$ end conditions at each side of the curve. 

In the following, we give the stencils for the cases $k=2, 3$. 

\paragraph{Quadratic curves ($k=2$).}
We have $\Pi_2=(P_0,P_1,P_2)$ with initial node vector $(0 0 0 1 1 1)$.
The polygons $\Pi^1_{LR}=(P_0^1,P_1^1,P_2^1,P_3^1)$ and $\Pi_{LR}^2=(P_0^2,...,P_5^2)$ at the first two levels of subdivision are obtained with the special rules
\[
\begin{array}{c}
P_0^1=P_0,\hspace{0.2cm} P_1^1=\frac{1}{2}P_0+\frac{1}{2}P_1,\hspace{0.2cm} P_2^1=\frac{1}{2}P_1+\frac{1}{2}P_2,\hspace{0.2cm} P_3^1=P_2\\[5pt]
\mbox{and}\\[5pt]
P_0^2=P_0^1,\hspace{0.3cm} 
P_1^2=\frac{1}{2}P_0^1+\frac{1}{2}P_1^1,\hspace{0.3cm} 
P_2^2=\frac{3}{4}P_1^1+\frac{1}{4}P_2^1,\\[5pt]
P_3^2=\frac{1}{4}P_1^1+\frac{3}{4}P_2^1,\hspace{0.3cm}
P_4^2=\frac{1}{2}P_2^1+\frac{1}{2}P_3^1,\hspace{0.3cm} 
P_5^2=P_3^1
\end{array} 
\]
From the second level on, we apply the following general rules:
\begin{equation}
\label{eq:LR2}
\begin{array}{l}
\hspace{-6pt} \left. \begin{array}{l}
P_{2j}^{n+1}=\frac{3}{4}P_j^n+\frac{1}{4}P_{j+1}^n\\[5pt]
P_{2j+1}^{n+1}=\frac{1}{4}P_j^n+\frac{3}{4}P_{j+1}^n\\[5pt]
\end{array} \hspace{0.5cm} \right\}  j=1...2^n-1\\[12pt]
P_0^{n+1}=P_0^n\\[5pt]
P_1^{n+1}=\frac{1}{2}P_0^n+\frac{1}{2}P_1^n \\[5pt]
P_{2^{n+1}} ^{n+1}=\frac{1}{2}P_{2^n}^n+\frac{1}{2}P_{2^n+1}^n\\[5pt]
P_{2^{n+1}+1} ^{n+1}=P_{2^n+1}^n
\end{array}
\end{equation} 
where the first two rows are the standard stencils of the Lane-Riesenfeld subdivision of order 1 (which coincides with the Chaikin subdivision in the uniform case).
Note that all stencils are affine averages of at most two points. 

\paragraph{Cubic curves ($k=3$).}
We have $\Pi_3=(P_0,P_1,P_2,P_3)$ and initial node vector $(0 0 0 0 1 1 1 1)$.
The polygons at the first and second levels are obtained with the special rules
$$P_0^1=P_0,\hspace{0.1cm} 
P_1^1=\frac{1}{2}P_0+\frac{1}{2}P_1,\hspace{0.1cm} 
P_2^1=\frac{1}{2}P_1+\frac{1}{2}P_2,\hspace{0.1cm} 
P_3^1=\frac{1}{2}P_2+\frac{1}{2}P_3,\hspace{0.1cm} 
P_4^1=P_3$$ 
and 
\[
\begin{array}{c}
P_0^2=P_0^1,\hspace{0.3cm} 
P_1^2=\frac{1}{2}P_0^1+\frac{1}{2}P_1^1,\hspace{0.3cm} 
P_2^2=\frac{3}{4}P_1^1+\frac{1}{4}P_2^1,\\[5pt]
P_3^2=\frac{3}{16}P_1^1+\frac{5}{8}P_2^1+\frac{3}{16}P_3^1,\\[5pt]
P_4^2=\frac{1}{4}P_2^1+\frac{3}{4}P_3^1,\hspace{0.3cm} 
P_5^2=\frac{1}{2}P_3+\frac{1}{2}P_4,\hspace{0.3cm} 
P_6^2=P_4^1
\end{array} 
\]

From the second level on, we apply the following general rules:
\begin{equation}
\label{eq:LR}
\begin{array}{l}
\hspace{-6pt} \begin{array}{ll}
P_{2j}^{n+1}=\frac{1}{2}P_j^n+\frac{1}{2}P_{j+1}^{n} & j=2...2^{n}-2\\[5pt]
P_{2j+1}^{n+1}=\frac{1}{8}P_j^{n+1}+\frac{3}{4}P_{j+1}^{n+1}+\frac{1}{8}P_{j+2}^{n+1} &  j=2...2^{n}-3\\[5pt]
\end{array} \\[12pt]
P_0^{n+1}=P_0 \\[5pt]
P_1^{n+1}=\frac{1}{2}P_0^{n}+\frac{1}{2}P_1^{n} \\[5pt]
P_2^{n+1}=\frac{3}{4}P_1^{n}+\frac{1}{4}P_2^{n} \\[5pt]
P_3^{n+1}=\frac{3}{16}P_1^{n}+\frac{11}{16}P_2^{n}+\frac{2}{16}P_3^{n} 
\end{array}
\end{equation}
where, for the sake of brevity, we have omitted the end conditions to the right end side, since they are symmetric 
to the end conditions to the left end side. 

In order to apply such scheme in the manifold case, we recall that the RCM is not well defined for three or more points, unless they are all contained in a convex set.
We overcome this limitation by factorizing the weighted averages of three points as repeated averages between pairs of points, 
to be computed in terms of the operator $\mathcal{A}$.
Note that, while in the Euclidean case the result is independent of the factorization, in the manifold case a different factorization yields a different curve in general.

Inspired from previous literature, we define our scheme in the manifold setting based on the following factorization. 
The average defining $P_{2j+1}^{n+1}$ is expressed as one step of midpoint subdivision and  two steps of smoothing, as prescribed by the uniform LR scheme, which has been already investigated in the manifold setting \cite{Duchamp:2018eu}. 
This can be thus compactly written as
\[
\begin{array}{c}
Q_{2j}^{n+1} = \frac{1}{4}P_j^{n+1}+\frac{3}{4}P_{j+1}^{n+1},\hspace{0.3cm} Q_{2j+1}^{n+1} = \frac{3}{4}P_{j+1}^{n+1}+\frac{1}{4}P_{j+2}^{n+1}\\[5pt] 
P_{2j+1}^{n+1}=\frac{1}{2}Q_{2j}^{n+1}+\frac{1}{2}Q_{2j+1}^{n+1}.
\end{array}
\]
For $P_2^3$ and $P_3^{n+1}$ instead, we use the \textit{inductive means} proposed in \cite{Dyn:2017fr}, which sort the terms by their weights and average the points with the largest weights first.
We thus obtain:
$$ Q = \frac{3}{13}P_1^1+\frac{10}{13}P_2^1,\hspace{0.2cm} P_3^2=\frac{10}{16}Q+\frac{3}{16}P_3^1.$$
and
\[
\begin{array}{c}
R_{1}^{n+1} = \frac{3}{14}P_1^{n}+\frac{11}{14}P_{2}^{n},\hspace{0.3cm}
P_{3}^{n+1}=\frac{14}{16}R_{1}^{n+1}+\frac{2}{16}P_{3}^{n+1}.
\end{array}
\]
We extend this scheme to the manifold setting, by substituting 
each affine average with the corresponding application of the manifold average 
$\mathcal{A}$. We omit the details for the sake of brevity. One step of subdivision for $n=3$ is exemplified, in the manifold setting, 
in Figure~\ref{fig:RDC-OLR} (OLR).
Note that $n=3$ is the first level in which the stencils of the uniform LR subdivision apply.

The same approach applies in the case $k>3$ too, where the $2k-2$ special stencils at the boundaries can be computed through the generalized Oslo algorithm \cite{Lyche:93}, and the remaining (internal) points are obtained by applying the classical LR algorithm.
Factorization of operations as repeated averages of pairs of points applies for all values of $k$ by exploiting the nature of the LR scheme, which consists of one step of midpoint subdivision, followed by $k-1$ steps of smoothing by averaging (see, e.g., \cite{Goldman:2002}, Chapter 7).

\section{Proofs of convergence and smoothness of RDC and OLR}
\label{app:proofs}
We prove the propositions of Section \ref{sec:method} concerning the convergence and smoothness of the subdivision schemes RDC and OLR. 
Both proofs rely on previous results ``in the small'', plus the fact that both schemes have the \emph{contractivity property}, i.e., the lengths of segments in the control polygon shrink at each iteration for a constant factor. 


\begin{lemma}
\label{rad_min_ball}
Let $\ell$ be the total length of a geodesic polyline $\Pi = (p_0,...,p_k)$, i.e.
$$\ell=\sum \limits_{i=0}^{k-1}d(p_i,p_{i+1})$$
Then the radius of the minimal enclosing ball of  $\{p_0,..,p_k\}$ is not greater than $\frac{\ell}{2}$.
\end{lemma}
\begin{proof}
Let $\gamma_i(t)$ denote the geodesic from $p_i$ to $p_{i+1}$, $i=0,...,k-1$, and let us denote $L_i$ the length of  $\gamma_i(t)$. 
Let $o$ be the midpoint of $\Pi$, i.e., the point that has equal distance to $p_0$ and $p_k$ when distance is measured along $\Pi$.
And let $h \in \{0,...,k-1\}$ be such  that $o \in \gamma_h(t)$. 
Then we have 
\begin{equation}
\label{def_center}
\frac{\ell}{2}=\sum \limits_{j=0}^{h-1} L_j+d(\gamma_{h}(0),o)=\sum \limits_{j=h+1}^{k-1} L_j+d(\gamma_{h}(1),o).
\end{equation}
Then by triangular inequality we have: 
\begin{itemize}[noitemsep,topsep=0pt]
\item if $i\leq h$, then $d(p_i,o)\leq \sum \limits_{j=i}^{h-1} L_j+d(\gamma_{h}(0),o)\leq\frac{\ell}{2}$,
\item if $i> h$, then  $d(p_i,o)\leq \sum \limits_{j=h+1}^{i} L_j+d(\gamma_{h}(1),o)\leq\frac{\ell}{2}$,
\end{itemize} 
Therefore, $p_i \in \overline{B(o, \ell/2)}$ for every  $i\in 0,\ldots,k$. 

\end{proof}
From now on, 
$\Pi^n=(p_0,\ldots,p_m)$ will denote the control polygon $\Pi^n$ obtained after $n$ iterations of a given subdivision method. 
Moreover, we define $L_j^n=d(p_j,p_{j+1})$ for every $j \in 0,\ldots,m-1$ and
$$L^n:=\max \limits_{j \in 0,\ldots,m-1} L^n_j.$$

\begin{definition}
A given subdivision method satisfies the \emph{contractivity property} if there exists $\alpha \in (0,1)$ such that for every $n\in \mathbb{N}$ we have $L^{n}\leq \alpha^{n} L^0.$
\end{definition}

 \begin{lemma}
 \label{shrinking_RDC}
The geodesic RDC scheme satisfies the contractivity property with $\alpha=\frac{1}{2}$.
\end{lemma}
\begin{proof}
Let $\Pi^0=\{p_0,...,p_k\}$ be the initial control polygon. 
Since RDC acts on every $k+1$ consecutive control points independently at each level of recursion, w.l.o.g., we consider just the first iteration, and show that $L^1\leq \frac{1}{2} L^0$. 
The polygon $\Pi^1$ subdividing $\Pi^0$ is the concatenation of two sub-polygons $\Pi_L$ and $\Pi_R$, joined at point $\mathbf{b}_0^k(\frac{1}{2})$ according to Eq.~\ref{eq:decasteljau}. 
We show the contractivity just for $\Pi_L$, the case of $\Pi_R$ being symmetric. 
Since all $\mathbf b_i^r(t)$ in Eq.~\ref{eq:decasteljau} are evaluated for $t=\frac{1}{2}$ in the RDC scheme, for the sake of brevity we omit the argument and denote them simply with $\mathbf b_i^r$.

We first show that all segments of the intermediate polygons involved in the construction of  Eq.~\ref{eq:decasteljau} are not longer than $L^0$.
For all $r\in 1\ldots k-1$ and all $i\in 0\ldots k-r-1$, by construction, the segment $\mathbf{b}_i^r \mathbf{b}_{i+1}^r$ is a geodesic line joining the midpoints of segments 
$\mathbf{b}_i^{r-1} \mathbf{b}_{i+1}^{r-1}$ and $\mathbf{b}_{i+1}^{r-1} \mathbf{b}_{i+2}^{r-1}$.
It is straightforward to see that 
$$d(\mathbf{b}_i^r,\mathbf{b}_{i+1}^r) \leq \frac{1}{2} d(\mathbf{b}_i^{r-1},\mathbf{b}_{i+1}^{r-1})
+ \frac{1}{2} d(\mathbf{b}_{i+1}^{r-1},\mathbf{b}_{i+2}^{r-1}) \leq L^0,$$
where the first inequality follows from the triangular inequality, while the second follows by induction on $r$.
From Eq.~\ref{eq:decasteljau} we have 
$$\Pi_L=(p_0=\mathbf{b}_0^0,\mathbf{b}_0^1,\ldots,\mathbf{b}_0^k).$$
The length $d(\mathbf{b}_0^r,\mathbf{b}_0^{r+1})$ of each segment of $\Pi_L$ is, by construction, half the length of segment 
$\mathbf{b}_0^r,\mathbf{b}_1^r$, hence not greater than $\frac{L^0}{2}$.
\end{proof}
\begin{lemma}
 \label{shrinking_LR}
The geodesic LR scheme satisfies the contractivity property with $\alpha=\frac{1}{2}$.
\end{lemma}
\begin{proof}
By definition, for a B-spline curve of degree $k$, one level of LR subdivision consists of one step of midpoint subdivision, followed by $k-1$ smoothing steps.
We show the contractivity factor by induction on the number of smoothing steps (averages between consecutive points). 
To simplify the notation, we omit the dependence of the control points from the iteration, and we omit the dependence of the average operator $\mathcal A$ from weight $w$, since from now on we fix $w=\frac{1}{2}$. 
Let us define
$$q_{2i}^0=p_i^n, \hspace{0.3cm} q_{2i+1}^0=\mathcal{A}(p_i^n,p_{i+1}^n),$$
where $p_i^n$ denotes the $i$-th control points of $\Pi_{LR}^n$ and $q_j^0$ the points obtained with the midpoint subdivision step, and
$$q_{2i}^r=\mathcal{A}(q_{2i}^r,q_{2i+1}^r), \hspace{0.3cm}q_{2i+1}^r=\mathcal{A}(q_{2i+1}^r,q_{2i+2}^r),$$
where the $q_j^r$ are obtained with the $r$-th smoothing step.
If $r=1$, by definition we have
\begin{align*}
d(q_{2i}^1,q_{2i+1}^1)&=d(\mathcal{A}(q_{2i}^0,q_{2i+1}^0),\mathcal{A}(q_{2i+1}^0,q_{2i+2}^0))\\
&\leq d(\mathcal{A}(q_{2i}^0,q_{2i+1}^0),q_{2i+1}^0)+d(q_{2i+1}^0,\mathcal{A}(q_{2i+1}^0,q_{2i+2}^0))\\
&=\frac{L_i^n}{4}+\frac{L_i^n}{4}\\
&\leq \frac{L^n}{2},
\end{align*}
and by symmetry we have that $d(q_{2i+1}^1,q_{2i+2}^1)\leq \frac{L^n}{2}$.
Let us show now that if the statement holds for $r$, then it holds for $r+1$.
By definition we have that 
\begin{align*}
d(q_{2i}^{r+1},q_{2i+1}^{r+1})&=d(\mathcal{A}(q_{2i}^r,q_{2i+1}^r),\mathcal{A}(q_{2i+1}^r,q_{2i+2}^r))\\
&\leq\frac{1}{2}(d(q_{2i}^r,q_{2i+1}^r)+d(q_{2i+1}^r,q_{2i+2}^r))\\
&\leq\frac{1}{2}(\frac{L^n}{2}+\frac{L^n}{2})\\
&\leq \frac{L^n}{2}
\end{align*}
Similarly it can be shown that $d(q_{2i+1}^{r+1},q_{2i+2}^{r+1})\leq L^n/2$. 
Since the above considerations do not depend on the iteration $n$, we have that LR satisfies the contractivity property with $\alpha=\frac{1}{2}$.
\end{proof}

\begin{lemma}
\label{shrinking}
If a subdivision method satisfies the contractivity property, then for every $\delta>0$ there exists $n_\delta \in \mathbb{N}$ such that after $n_\delta$ iterations of subdivision every consecutive $k+1$-uple of control points is contained in $\overline{B(o,\delta)}$ for some $o \in \mathcal{M}$, where $k+1$ is the number of initial control points .
\end{lemma}
\begin{proof}
Let $\Pi^n$ be the control polygon obtained after $n$ iterations and suppose that $\{q_0,....,q_m\}$ are the control points of $\Pi^n$. 
Then we define
$$L^n:=\max \limits_{i \in \{0,...,m-1\}}d(q_i,q_{i+1}).$$
Let us put $L:=L^0$ and let $n_\delta$ be defined as
$$n_\delta:= \lfloor \log_{\alpha}\frac{\delta}{\widetilde{k} L}\rfloor +1,$$
where $\widetilde{k}:=\lfloor \frac{k+1}{2}\rfloor +1$. By definition,  we have $(k+1)/\tilde{k}<2$. 
Hence
$$L^{n_\delta}\leq \alpha^{n_\delta} L<\frac{\delta}{\widetilde{k}}$$
It follows that every sub-polyline $\Gamma_i^{n_\delta}(p_i,...,p_{i+k})$ of $\Pi^{n_\delta}$, which is defined by the control points $(p_i,...,p_{i+k})$,  has a total length $\ell$ (in the sense of Lemma \ref{rad_min_ball}) satisfying
$$\ell\leq (k+1)L^{n_\delta}<2\delta.$$
 By Lemma \ref{rad_min_ball} we conclude.

\end{proof}

\begin{proof}[Proof of Proposition \ref{sub-smoothness-RDC}]
Noakes proved that, for $k=2,3$, if the initial control points are contained in a strongly convex ball, then RDC converges to a $C^1$ curve interpolating the control polygon and being tangent to it at its endpoints \cite{Noakes:1998bo,Noakes:1999er}. 
Let us define
$$\delta=\inf\limits_{p \in \mathcal{M}} \{r>0 : B(p,r) \text{ is strongly convex}\},$$
which is called the \textit{convexity radius} of $\mathcal{M}$. 
It is well known that if $\mathcal{M}$ has bounded sectional curvature, then $\delta>0$ \cite{Sakai1997}. 

By Lemma \ref{shrinking_RDC} and Lemma \ref{shrinking}, we know that we can choose $n_\delta$ such that $\Pi_{DC}^{n_\delta}$ consists of a sequence of control polygons of order $k$, each contained in a ball of radius not greater than $\delta$. 
This means that for $n>n_\delta$ every control polygon in $\Pi_{DC}^n$ converges to a $C^1$ curve as above. 
Furthermore, by construction, every two segments incident at a junction point in  $\Pi_{DC}^n$ are branches of the same geodesic line.
Condition (ii) of Theorem 3 in \cite{Popiel:2007bt} warrants that $C^1$ continuity is satisfied at all junction points, too. 
\paragraph{Sketch of proof for the case $k>3$}
As shown in  \cite{Noakes:1998bo}, the contractivity property can be used to show that the RDC scheme satisfies some proximity conditions that lead to the $C^1$ continuity of the limit curve for $k=2,3$. The arguments on which the proofs rely are extendable to the case of an arbitrary $k$ in the following sense: the number of control points does not affect the convergence of the method or the study of the smoothness of the limit curve. Indeed, the main condition that must hold is the contractivity property, which is satisfied for every $k$ as shown in Lemma \ref{shrinking_RDC}. Hence, the computations that occur in Lemmas 4.4-4.6 of \cite{Noakes:1998bo}  can be made in the case of a generic $k$ by applying the results on Taylor's expansion of Section 4 of that paper, and to express every control point of $\Pi_{DC}^{n+1}$ in terms of the control points of  $\Pi_{DC}^{n}$.
\end{proof}

\begin{proof}[Proof of Proposition \ref{sub-smoothness-OLR}]
The smoothness of the \textit{uniform geodesic} LR scheme has been recently investigated \cite{Duchamp:2018eu}. 
It has been proved that if all the $k+1$ control points are contained in a totally normal neighborhood, then the limiting curve is $C^{k-1}$. As before, by Lemma \ref{shrinking_LR} and Lemma \ref{shrinking} we know that we can choose $n_\delta$ such that every consecutive $k+1$-uple in $\Pi_{LR}^{n_\delta}$ are contained in a totally normal neighborhood. This will give us a $C^{k-1}$ curve at every point obtained by applying the uniform geodesic LR scheme,  i.e., the ones "in the middle". 
More formally, let us fix $t \in (0,1)$, and suppose that, after $n_\delta$ iterations of the OLR algorithm $t\in [2^{-n_\delta}j,2^{-n_\delta}(j+1)]$ for some $j \in \mathbb{N}$. By the previous considerations and by definition of the OLR scheme, 
we know that if the $k+2$ knots $2^{-n_\delta}j,...,2^{-n_\delta}(j+k+1)$ are uniformly spaced, then the control polygon defined by $P_{j-k},P_{j-k+1},...,P_j$ undergoes to the uniform LR subdivision rules. 
Hence, we now that the B-spline segment $B_j(t)$ will be $C^{k-1}$ for every $t \in [2^{-n_\delta}j,2^{-n_\delta}(j+1)]$. 
Consider now the case where it is not possible to apply the uniform LR scheme to $P_{j-k},P_{j-k+1},...,P_j$. 
In this case, we observe that, by definition, $\Pi^{n_\delta +1}$ is obtained by subdividing every subinterval of the knot vector at iteration $n_\delta$ of the form  
$[2^{-n_\delta}j,2^{-n_\delta}(j+1)]$ by adding the knot $(2j+1)2^{-(n+1)}$. 
Since $t \neq 0,1$, there exists $n_t$ such that $t \in [2^{-n_t}j_t,2^{-n_t}(j_t+1)]$ for some $j_t\in \mathbb{N}$ and the related control polygon can be subdivided with the LR uniform stencils. 
Hence, for every $t \in (0,1)$, there is $\bar{n}:=\max(n_\delta,n_t)$ such that $ t \in [2^{-\bar{n}t}\bar{j},2^{-\bar{n}}(\bar{j}+1)]$ for some $\bar{j} \in \mathbb{N}$ and such that $B_{\bar{j}}(t)$ will be a $C^{k-1}$ B-spline segment.

Concerning the endpoints of $\Pi_{LR}^n$, i.e., for $t=0$ and $t=1$, the end conditions in Equations \ref{eq:LR2} and \ref{eq:LR} guarantee that the limit curve interpolates the endpoints of the initial control polygon $\Pi_k$, and it will be tangent  to the first and last segments of $\Pi_k$ at its endpoints. 
In fact, the boundary stencils force the second control point to lie on the geodesic $\gamma_0(t)$ connecting $p_0$ to $p_1$ at all levels of subdivision. 
More precisely, we have that $p_1^n=\gamma_0(2^{-n})$, which implies that
$$\frac{p_1^n-p_0^n}{2^n}=\frac{\gamma_0(2^{-n})-\gamma_0(0)}{2^n}\approx \gamma_0'(0) \hspace{0.2cm}\text{for a big enough}\hspace{0.1cm} n,$$
and similarly for $p_{k-1}^n$.
Hence, the curve obtained with this method is $C^{k-1}$, and provides enough control to weld it with $C^1$ continuity to other curves.
\end{proof}

\section{From B-spline to B\'ezier}
\label{app:BtoBezier}
\paragraph{Uniform case}
Given the control polygon $(P_0,...,P_k)$ of a uniform B-spline segment $\mathbf{b}(t)$ of degree $k$, 
the control polygon $(Q_0,...,Q_k)$ of a B\'ezier curve coincident with $\mathbf{b}(t)$ is
given by expression
%
%
%
%
$$(Q_0,Q_1,\ldots,Q_k)^T = M_b^{-1}M_s (P_0,P_1,\ldots,P_k)^T,$$
where $M_s$ and $M_b$ are the matrices defining the $k$-degree B-spline and the $k$-degree B\'ezier curve, respectively. 
Both matrices are well defined for an arbitrary degree $k$, and their construction can be found in \cite{Yamaguchi:1988}. 


As an example, for $k=3$, the above equation leads to
\[\begin{array}{ll}
Q_0=\frac{1}{6}(P_0+4P_1+P_2) &
Q_1=\frac{1}{6}(4P_1+2P_2)\\[5pt]
Q_2=\frac{1}{6}(2P_1+4P_2) &
Q_3=\frac{1}{6}(P_1+4P_2+P_3),
\end{array}
\]
where the first average (and similarly the last) can be factorized as
\[\begin{array}{c}
\tilde{Q}_0=\frac{1}{3}(P_0+2P_1) \hspace{0.5cm} \tilde{Q}_1=\frac{1}{3}(2P_1+P_2)  \hspace{0.5cm} 
Q_0=\frac{1}{2}(\tilde{Q}_0+\tilde{Q}_1).
\end{array}
\]

It is important to point out that the inverse computation of $(P_0,...,P_k)$ from $(Q_0,...,Q_k)$ leads to non-convex averages, i.e., with negative weights. 
This explains why, in our case, we consider an open-uniform LR subdivision, rather than computing the control polygon of a uniform B-spline, since this would have implied the extrapolation of long geodesic lines, which may become unstable. 

\paragraph{Non-uniform case} 
In this case, the entries of $M_s$ depend on the spacing between the entries of the knot vector of the B-spline. 
The construction of $M_s$ in the general case can be found in \cite{Qin:2000}; we report here, as an example, the case where the knot vector has the form $(00001234\ldots)$ and $\mathbf{b}(t)$ is the cubic B-spline segment defined in $[0,1]$. If $(P_0,P_1,P_2,P_3)$ are the control points defining $\mathbf{b}(t)$, then we have
\[\hspace{-2mm} \begin{array}{llll}
Q_0=P_0 & Q_1=P_1 &
Q_2=\frac{1}{2}(P_1+P_2) & Q_3=\frac{1}{12}(3P_1+7P_2+2P_3),
\end{array}\]
where again we can factorize
\[\begin{array}{c}
\tilde{Q}_0=\frac{1}{2}(P_0+P_1) \hspace{0.5cm} \tilde{Q}_1=\frac{1}{3}(2P_1+P_3)\hspace{0.5cm}  
Q_3=\frac{1}{2}(\tilde{Q}_0+\tilde{Q}_1).
\end{array}
\]

\section{Point Insertion on a Curve of Degree $k>3$}
\label{app:pointinsertion}

We describe here the algorithms for point insertion on a curve of degree $k>3$, for the RDC and the OLR schemes. 
The algorithms for the special cases $k=2$ and $k=3$ were already described in Sections 4.1.3 and 4.2.4 of the paper. 
We keep the same notations.

We start by considering the RDC scheme.
Tree descent and polygon split at a leaf node work exactly as described previously, except that we compute and record the control polygons of both children at each splitting step.
When reaching a leaf, we split its polygon at value $\bar{t}$, thus obtaining the two polygons $\Pi_L$ and $\Pi_R$, respectively. 

Now we need to compute two control polygons: one defining the concatenation portion of curve defined by $\Pi_L$ and of the portion curve preceding it to the left; and likewise for $\Pi_R$ and the curve following it to the right. 
The two control polygons can be processed independently, therefore we describe here only the case of $\Pi_L$. 

Since the whole curve is of degree $k$, for any two consecutive chunks of curve, we can find a control polygon of degree $k$ that describes the concatenation of such chunks.
In order to do that, we backtrack on the tree path that we descended, by chaining control polygons in pairs, until we reach the root. 
We build the polygon corresponding to the concatenation of two consecutive chunks by reversing the De Casteljau construction. 

Starting at $\Pi_L$, we climb the tree until we find a sibling $\Pi_{LL}$ to the left $\Pi_L$. 
Note that $\Pi_{LL}$ has been computed while descending the tree, and needs not be at the same level of $\Pi_L$.
We now build the polygon $\bar{\Pi}$ encompassing the concatenation of the two curves described by $\Pi_{LL}$ and $\Pi_L$.

\begin{figure}[t]
  \centering
  \includegraphics[width=0.8 \columnwidth]{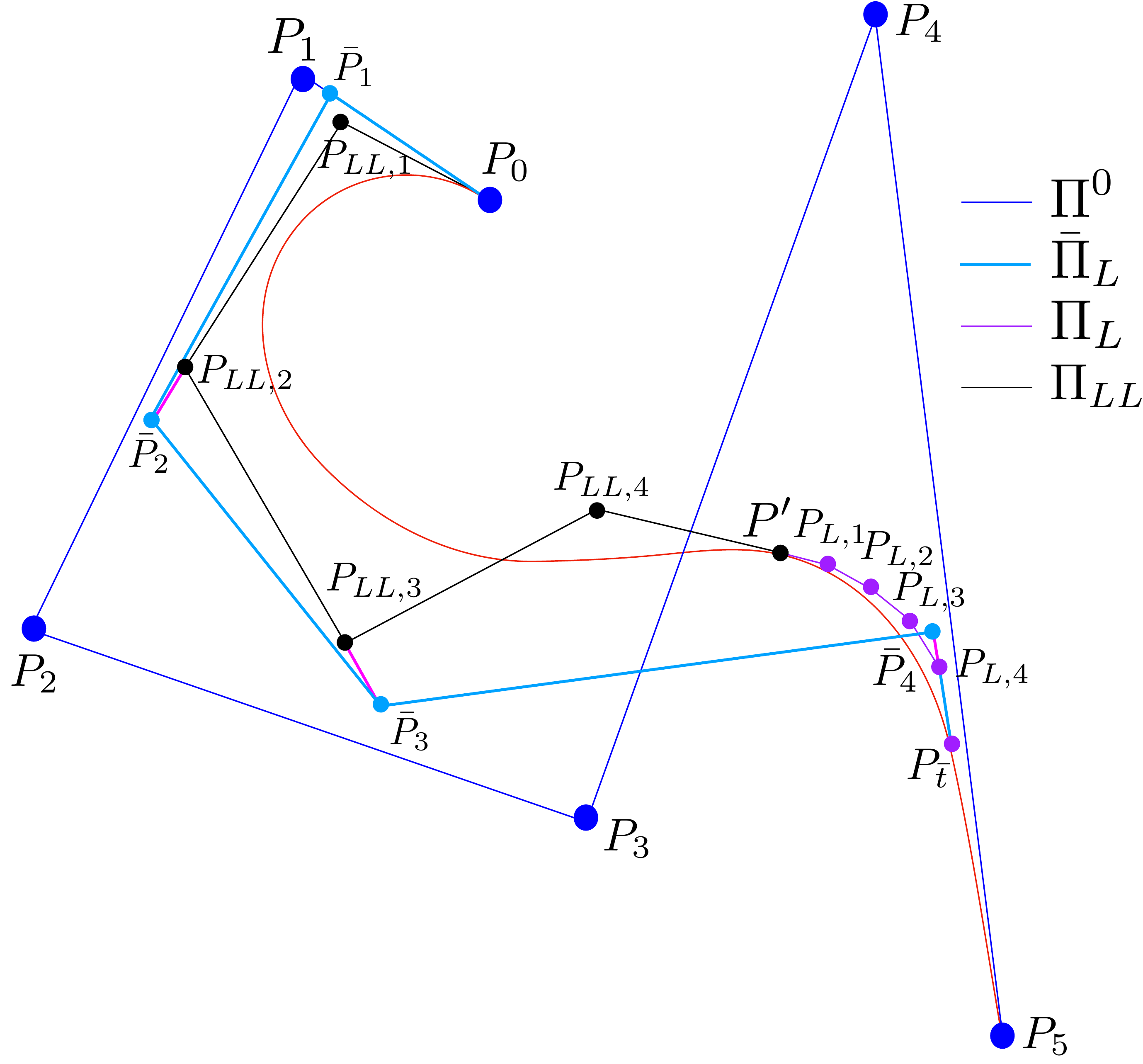}
  \caption{Point insertion on a quartic curve (left side only). The control polygon $\bar{\Pi}_L$ (light blue) encompassing the two curves described by two quintic polygons $\Pi_{LL}$ (black) and $\Pi_L$ (purple) is built by extending two segments of $\Pi_{LL}$ and one segment of $\Pi_L$ (extensions are depicted in magenta).}
  \label{fig:DC-insert-generic}
  \Description{Image}
\end{figure}

Let $\bar{t}$ be  the parameter that identifies $P_{\bar{t}}$ on the input curve, then our purpose is to determine the $k+1$ control points $\bar{P}_0,...,\bar{P}_{k}$ of $\bar{\Pi}_L$ such that the curve defined by such control polygon (nearly) coincide with the portion of the input curve in $[0,\bar{t}]$. Since $\bar{P}_0=P_0$, $\bar{P}_k=P_{\bar{t}}$ and $\bar{P}_1=\gamma_(P_0,P_1,\bar{t})$, we just need to find the $k-2$ points $\bar{P}_2,...,\bar{P}_{k-1}$ through geodesic extensions. 
To do that, we need to properly extend one side of $\bar{\Pi}_L$ (to find $\bar{P}_{k-1}$) and $k-3$ sides of the control polygon $\Pi_{LL}$ (to find the other $k-3$ points).  
Note that if $k=2$, there are no points that need to be found by extending the sides of the control polygons $\Pi_{LL}$ and $\Pi_{L}$,  whereas if $k=3$ we just need to extend one side of $\Pi_{L}$, consistently with the algorithm already described in the paper.

 Let $P^\prime$ be the junction point of $\Pi_{LL}$ and $\Pi_L$, and let $t^\prime$ be its  parameter on the input curve. 
 By considering the reparametrization of point $P^\prime$ w.r.t.\ $\bar{\Pi}_L$ we must have
\begin{equation}
\label{eq:condition}
 d(P_{LL,i},P_{LL,i+1})=\frac{t^\prime}{\bar{t}}d(P_{LL,i}, \bar{P}_{i+1}),\hspace{0.5cm} i=1,..,k-3,
 \end{equation}
which allows us to determine the position of $\bar{P}_{i+1}$ for $i=1,...,k-3$. In fact, $\bar{P}_{i+1}$ is the endpoint of the geodesic obtained by extending $\gamma(P_{LL,i},P_{LL,i+1},t)$ for a length $\delta_i$, where
$$\delta_i=\bigg(\frac{\bar{t}-t^\prime}{t^\prime}\bigg) \cdot d(P_{LL,i}, P_{LL,i+1}),\hspace{0.5cm} i=1,..,k-3.$$

To determine $\bar{P}_{k-1}$, we proceed as described previously, i.e 
by extending $\gamma(P_{\bar{t}},P_{L,k-1},t)$ for a length $d(P_{\bar{t}},P_{L,k-1})\cdot (t^{\prime}/(\bar{t}-t^{\prime})).$ 

The procedure described above can be applied to the OLR scheme too. 
After the tree descent, we obtain the control polygon $\widetilde{\Pi}$ of the B-spline segment $B(t_n)$, which defines the input curve in $ [t_{n},t_{n+1}]$, for some $n \in \mathbb{N}$, where $\bar{t }\in [t_{n},t_{n+1}]$. 
As shown in Appendix C, we can obtain the control polygon $\Pi$ of the B\'ezier curve coincident to $B(t_n)$ from $\widetilde{\Pi}$ through manifold averages.
At this point, we proceed as in the RDC case, by splitting $\Pi$ in two polygons $\Pi_L$ and $\Pi_R$ and by applying the algorithm described above.
The conversion from B-spline to B\'ezier is applied at each sibling we find, which must be concatenated with the  polygon at hand. 

\bibliographystyle{ACM-Reference-Format}
\bibliography{references}

\end{document}